\documentclass[conference]{IEEEtran}

\usepackage[T1]{fontenc}
\usepackage{microtype}

\usepackage{amsmath}
\usepackage{amssymb}
\usepackage{amsthm}
\usepackage{mathtools}
\usepackage{verbatim} 
\usepackage{stmaryrd}
\usepackage{url}
\allowdisplaybreaks
\usepackage{basic}
\usepackage{ambients}
\usepackage{paralist}
\usepackage{graphicx}
\usepackage{proof}
\usepackage{hyperref}

\usepackage{enumitem}

\usepackage{epsf}
\usepackage{graphics} 
\usepackage{epsfig} 

\newcommand{\env}{\mathit{Env}}

\newif\ifdraft\drafttrue

\newcommand{\ntrans}[1]{\mathrel{{\trans{#1}}\makebox[0em][r]{$\not$\hspace{2ex}}}{\!}}
\newcommand{\ttrans}[1]{\stackrel{\, {#1} \,}{\Longrightarrow}}
\newcommand{\ttranst}[1]{\Longrightarrow\trans{#1}\Longrightarrow}

\newcommand{\on}{\mathsf{on}}
\newcommand{\off}{\mathsf{off}}

\newcommand{\bool}[1]{\llbracket #1 \rrbracket}

\newcommand{\fixn}[1]{\textbf{FIX}\footnote{#1}}
\renewcommand{\marginpar}[1]{\fixn{#1}}
\newcommand{\confCPS}[2]{#1 \, {\Join} \, #2}
\newcommand{\confCPSS}[3]{#1 \, {\Join_{#2}} \, #3}

\newcommand{\rsens}[2]{\mathsf{read}\, #2(#1)}

\newcommand{\wact}[2]{\mathsf{write}\, #2 \langle #1 \rangle}

\newcommand{\todob}[1]{{\color{blue} #1}}

\newcommand{\mm}[1]{\todob{#1}}

\newcommand{\statefun}{\xi_{\mathrm{x}}}
\newcommand{\actuatorfun}{\xi_{\mathrm{u}}}
\newcommand{\uncertaintyfun}{\xi_{\mathrm{w}}}
\newcommand{\evolmap}{\mathit{evol}}
\newcommand{\errorfun}{\xi_{\mathrm{e}}}
\newcommand{\measmap}{\mathit{meas}}
\newcommand{\invariantfun}{\mathit{inv}}
\newcommand{\safefun}{\mathit{safe}}

\newcommand{\CPS}{CPS}

\newtheorem{definition}{Definition}
\newtheorem{theorem}{Theorem}
\newtheorem{proposition}{Proposition}
\newtheorem{remark}{Remark}
\newtheorem{lemma}{Lemma}
\newtheorem{example}{Example}
\newtheorem{corollary}{Corollary}

\usepackage{marvosym} 

\begin{document}





%


\title{A Formal Approach to Cyber-Physical Attacks} 

 \author{\IEEEauthorblockN{Ruggero Lanotte} 
\IEEEauthorblockA{Dipartimento di Scienza e Alta Tecnologia \\ Universit\`a dell'Insubria, Como, Italy}
\and
\IEEEauthorblockN{Massimo Merro and Riccardo Muradore}
\IEEEauthorblockA{Dipartimento di Informatica \\ Universit\`a degli Studi di Verona, Italy}
\and
\IEEEauthorblockN{Luca Vigan\`o}
\IEEEauthorblockA{
 Department of Informatics\\ King's College London, UK
 }%
}
        
\maketitle

\begin{abstract}

We apply formal methods to lay and streamline theoretical foundations to reason about Cyber-Physical Systems (CPSs) and cyber-physical attacks.
We focus on 
integrity and DoS attacks to sensors and actuators of CPSs, 
and on the timing aspects of these attacks. Our contributions are threefold: (1) we define a hybrid process calculus to model both CPSs and cyber-physical attacks; (2) 
we define a threat model of cyber-physical attacks and provide the means to assess attack tolerance/vulnerability with respect to a given attack; (3)~we formalise how to estimate the impact of a successful 
 attack on a CPS and investigate possible quantifications of the success chances of an attack. We illustrate definitions and results by means of a non-trivial engineering application. 
\end{abstract}


\section{Introduction}

\subsubsection*{Context and motivation}
\emph{Cyber-Physical Systems (CPSs)} are integrations of networking and
distributed computing systems with physical processes that monitor and
control entities in a physical environment, with feedback loops where
physical processes affect computations and vice versa. For example, in
real-time control systems, a hierarchy of \emph{sensors}, \emph{actuators}
and \emph{control processing components} are connected to control
stations. Different kinds of \CPS{s} include
supervisory control and data acquisition (SCADA), programmable logic
controllers (PLC) and distributed control systems.

In recent years there has been a dramatic increase in the number of
attacks to the security of cyber-physical and critical systems, e.g.,
manipulating sensor readings and, in general, influencing physical
processes to bring the system into a state desired by the attacker. Many
(in)famous examples have been so impressive to make the international
news, e.g.: the Stuxnet worm, which reprogrammed PLCs of nuclear
centrifuges in Iran~\cite{stuxnet} or  the attack on a sewage treatment
facility in Queensland, Australia, which manipulated the SCADA system to
release raw sewage into local rivers and parks~\cite{SlMi2007}.

As stated in~\cite{GGIKLW2015}, the concern for consequences at the
physical level puts \emph{\CPS{} security} apart from standard
\emph{information security}, and demands for \textit{ad hoc} solutions to
properly address such novel research challenges.
The works that have taken up these challenges range from proposals of
different notions of cyber-physical security and attacks (e.g.,
\cite{BuMaCh2012,GGIKLW2015,KrCa2013}, to name a few) to pioneering
extensions to \CPS{} security of standard formal approaches (e.g.,
\cite{BuMaCh2012,Cardenas2015,Vigo2015}).
 However, to the best of our knowledge,
a systematic formal approach to cyber-physical attacks is still to be
fully developed.

\subsubsection*{Background}
The dynamic behaviour of the \emph{physical plant} of a \CPS{} is often 
represented by means of a \emph{discrete-time state-space
model\/}
consisting of two equations of the form
\begin{displaymath}
\begin{array}{rcl}
x_{k+1} = Ax_{k} + Bu_{k} + w_{k} &  \mathrm{and} &
y_k = Cx_{k} + e_k\,,
\end{array}
\end{displaymath}
where
$x_k \in \mathbb{R}^n$ is the current \emph{(physical) state}, $u_k \in
\mathbb{R}^m$ is the \emph{input} (i.e., the control actions implemented
through actuators) and $y_k \in \mathbb{R}^p$ is the \emph{output} (i.e.,
the measurements from the sensors). 
The \emph{uncertainty} $w_k \in \mathbb{R}^n$ and the \emph{measurement error} $e_k \in \mathbb{R}^p$ represent perturbation and sensor noise, respectively, 
and $A$, $B$, and $C$ 
are matrices modelling the dynamics of the physical system. Here, the
\emph{next state} $x_{k+1}$ depends on the current state $x_k$ and the
corresponding control actions $u_k$, at the sampling instant $k \in
\mathbb{N}$. The state $x_k$ cannot be directly observed: only its
measurements $y_k$ can be observed.

The physical plant is supported by a communication network through which
the sensor measurements and actuator data are exchanged with 
controller(s) and supervisor(s) (e.g., IDSs), 
which are 
the \emph{cyber} components (also called \emph{logics}) of a CPS.

\subsubsection*{Contributions}
In this paper, we focus on a formal treatment of both \emph{integrity} and
\emph{Denial of Service (DoS)} attacks to \emph{physical devices} (sensors
and actuators) of \CPS{s}, paying particular attention to the \emph{timing
aspects} of these attacks. The overall goal of the paper is to apply
formal methodologies to lay \emph{theoretical foundations} to reason about
and statically detect attacks to physical devices of \CPS{s}.

Our contributions are threefold. The first contribution is the definition
of a \emph{hybrid process calculus}, called \cname{}, to formally specify
both \CPS{s} and cyber-physical attacks. In \cname{}, \CPS{s} have two
components: a \emph{physical component} denoting the \emph{physical plant}
(also called environment) of the system, and containing information on
state variables, actuators, sensors, evolution law, etc., and a
\emph{cyber component} that governs access to sensors and actuators,
channel-based communication with other cyber components. Thus, channels
are used for logical interactions between cyber components, whereas
sensors and actuators make possible the interaction between cyber and
physical components.

\cname{} adopts a \emph{discrete notion of time}~\cite{HR95} and it is
equipped with a \emph{labelled transition semantics (LTS)} that allows us
to observe both \emph{physical events} (system deadlock and violations of
safety conditions) and \emph{cyber events} (channel communications). Based
on our LTS, we define two trace-based system preorders: a \emph{trace
preorder}, $\sqsubseteq$, and a \emph{timed variant},
$\sqsubseteq_{m..n}$, for $m, n \in \mathbb{N}^{+}\cup \infty$, which
takes into account discrepancies of execution traces within the time
interval $m..n$.

As a second contribution, we formalise a \emph{threat model} that specifies  attacks that can manipulate sensor and/or actuator signals in order to drive a \CPS{} into an undesired state~\cite{TeShSaJo2015}. 
Cyber-physical attacks typically tamper with both the physical (sensors and actuators) and the cyber layer. In our threat model, communication cannot be manipulated by the attacker, who instead may compromise 
(unsecured) physical devices, which is our focus. 
As depicted in \autoref{fig:threat-model}, our attacks may affect directly 
the sensor measurements or the controller commands. 
\begin{itemize}[noitemsep]
\item \emph{Attacks on sensors} consist of reading and eventually 
replacing $y_k$ (the sensor measurements) with $y^a_k$. 
\item \emph{Attacks on actuators} consist of reading, eavesdropping and eventually replacing the controller commands $u_k$ with $u^a_k$, affecting directly the actions the actuators may execute.
\end{itemize}
We group attacks into classes. A class of attacks takes into account both the malicious activity  $\I$ on physical devices and the \emph{timing parameters} $m$ and $n$ of the attack: begin and end of the attack. We represent a class $C$  as a total function $C \in [\I \rightarrow {\cal P}(m..n)]$. Intuitively, for $\iota \in \I$, $C(\iota) \subseteq m..n$ denotes the set of time instants when an attack of class $C$ may tamper with the device $\iota$.
As observed in~\cite{KrCa2013}, timing is a critical issue in CPSs because
the physical state of a system changes continuously over time, and as the
system evolves in time, some states might be more vulnerable to attacks
than others. For example, an attack launched when the target state
variable reaches a local maximum (or minimum) may have a great impact on
the whole system behaviour~\cite{BestTime2014}. Furthermore, not only the
timing of the attack but also the \emph{duration of the attack} is an
important parameter to be taken into consideration in order to achieve a
successful attack. For example, it may take minutes for a chemical reactor
to rupture~\cite{chemical-reactor}, hours to heat a tank of water or burn
out a motor, and days to destroy centrifuges~\cite{stuxnet}.

\begin{figure}[t]
\centering
\includegraphics[width=7cm,keepaspectratio=true,angle=0]{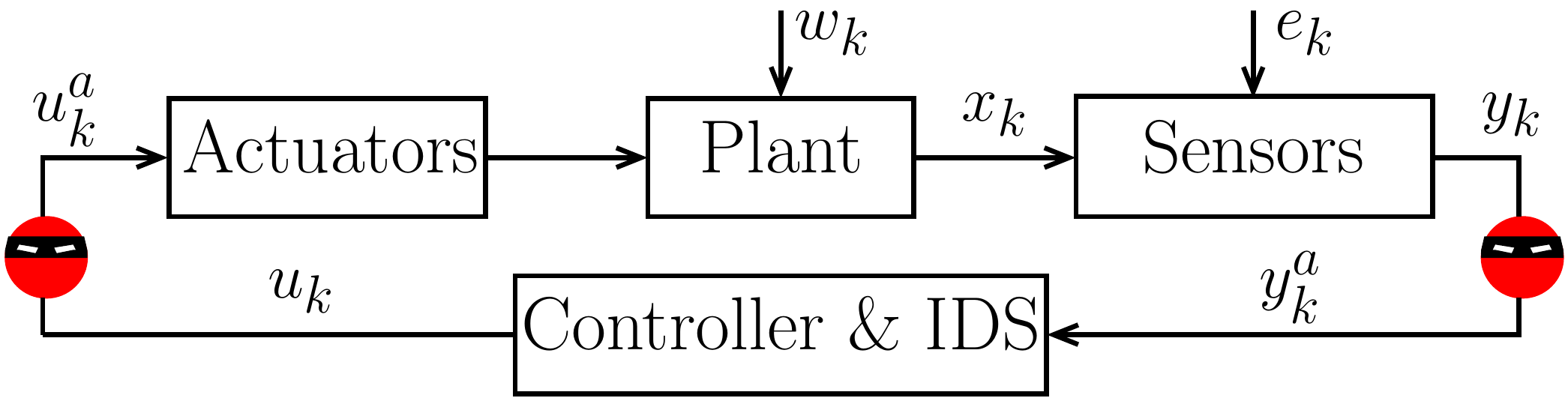}
\caption{Our threat model for CPSs}
\label{fig:threat-model}
\vspace*{-0.2cm}
\end{figure}
 
In order to make security assessments on our \CPS{s}, we adopt a
well-known approach called \emph{Generalized Non Deducibility on
Composition (GNDC)}~\cite{FM99}. 
Thus, in \cname{}, we say that a \CPS{} $\mathit{Sys}$
\emph{tolerates} a cyber-physical
attack $A$ if
\begin{displaymath}
\mathit{Sys} \: \parallel \: A \q \sqsubseteq \q \mathit{Sys} \enspace .
\end{displaymath}
In this case, the presence of the attack $A$, 
does not affect the whole (physical and logical) observable behaviour of the system $\mathit{Sys}$, and the attack can be considered harmless.

On the other hand, we say that a \CPS{} $\mathit{Sys}$
is \emph{vulnerable} to a cyber-physical attack $A$ of class 
$C \in [\I \rightarrow {\cal P}(m..n)]$ if there is a time interval $m'..n'$ in which the attack becomes observable (obviously, $m' \geq m$). Formally, we write: 
\begin{displaymath}
\mathit{Sys} \: \parallel \: A \q \sqsubseteq_{m'..n'} \q \mathit{Sys} \enspace .
\end{displaymath}

We provide sufficient criteria to prove attack tolerance/vulnerability to
attacks of an arbitrary class $C$. We define a notion of ``most powerful
attack'' of a given class $C$, $\mathit{Top}(C)$, and prove that if a
\CPS{} tolerates $\mathit{Top}(C)$ then it tolerates all attacks $A$ of
class $C$. Similarly, if a \CPS{} is vulnerable to $\mathit{Top}(C)$, in
the time interval $m'..n'$, then no attacks of class $C$ can affect the
system out of that time interval. This is very useful when checking for
attack tolerance/vulnerability with respect to all attacks of a given
class $C$.

As a third contribution, we formalise how to estimate the \emph{impact of
a successful attack on a \CPS{}} and investigate possible
\emph{quantifications of the chances} for an attack of being successful
when attacking a \CPS{}. 
This is important since, in industrial \CPS{s}, before taking any countermeasure against an attack, engineers typically first try to estimate the impact of the attack on the system functioning (e.g., performance and security) and weigh it against the cost of stopping the plant. If this cost is higher than the damage caused by the attack (as is sometimes the case),  then engineers might actually decide to let the system continue its activities even under attack.
We thus provide a \emph{metric} to estimate the
deviation of the system under attack with respect to expected behaviour,
according to its evolution law and the uncertainty of the model.
Then, we prove that the impact of the most powerful attack $\mathit{Top}(C)$ represents an upper bound for the impact of any attack $A$ of class $C$.


We introduce a non-trivial \emph{running example} taken from an
engineering application and use it to illustrate our definitions and cases
of \CPS{s} that tolerate certain attacks, and of \CPS{s} that suffer from
attacks that drag them towards undesired behaviours. 
We remark that while we have kept the example simple, it is actually far from trivial and designed to describe a wide number of attacks, as will become clear below.

All the results exhibited in the paper have been formally proved (due to lack of space, proofs are given in the appendix).
Moreover, the behaviour of our running example and of most of the cyber-physical attacks appearing in the paper have been simulated in MATLAB.

\subsubsection*{Organisation}
In \autoref{sec:calculus}, we give syntax and semantics of \cname{}. In \autoref{sec:cyber-physical-attackers}, we define cyber-physical attacks and provide sufficient criteria for attack tolerance/vulnerability. In \autoref{sec:impact}, we estimate the impact of  attacks on  \CPS{s}, and investigate possible quantifications of the success chances of an attack. In \autoref{sec:conclusions}, \nolinebreak we draw conclusions and \nolinebreak discuss \nolinebreak related
\nolinebreak and \nolinebreak future \nolinebreak work.


\section{The Calculus}
\label{sec:calculus}

In this section, we introduce our \emph{Calculus of Cyber-Physical Systems and Attacks}, \cname{}, which extends the \emph{Calculus of Cyber-Physical 
Systems} defined in~\cite{LaMe17} with specific features to formalise and study attacks to physical devices.




Let us start with some preliminary notations. 
We use 
$x, x_k \in \cal X$ for \emph{state variables},  
$c,d \in \cal C$ for \emph{communication channels}, 
$a, a_k \in \cal A$ for \emph{actuator devices}, 
 $s,s_k \in \cal S$ for \emph{sensors devices}, and 
$p,q$ for both sensors and actuators
(generically called \emph{physical devices}). 
\emph{Values}, ranged
over by $v,v' \in \cal V$, are built from basic values, such as
Booleans, integers and real numbers.
\emph{Actuator names} are metavariables for actuator devices like
$\mathit{valve}$, $\mathit{light}$, etc. Similarly, \emph{sensor names}
are metavariables for sensor devices, e.g., a sensor
$\mathit{thermometer}$. 

Given a 
 set of names $\cal N $, we write $\mathbb{R}^{\cal N} $ to denote the set of functions $[\cal N \rightarrow \mathbb{R} ]$ assigning a
real 
to each name in $\cal N$. For $\xi \in \mathbb{R} ^{\cal N}$,
$n \in \cal N$ and $v \in\mathbb{R} $, we write $\xi [n \mapsto v]$ for 
the function $\psi \in \mathbb{R} ^{\cal N}$ such that
$\psi(m)=\xi(m)$, for any $m \neq n$,  and  \ $\psi(n)=v$.

Finally, we distinguish between \emph{real intervals}, such as $(m, n]$, for 
$m \in \mathbb{R}$ and  $n \in \mathbb{R} \cup \infty$, and \emph{integer intervals}, written $m..n$, for $m \in \mathbb{N}$ and $n \in \mathbb{N} \cup \infty$. As we will adopt a discrete notion of time, we will use integer intervals to denote \emph{time intervals}.

\begin{definition}[Cyber-physical system]
In \cname{}, a \emph{cyber-physical system} consists of two components: a \emph{physical environment} $E$ that encloses all physical aspects of a system 
and a
\emph{cyber component\/} 
$P$ that interacts with sensors and actuators of the system, and can communicate, via channels, with other cyber components of the same  or of other \CPS{s}. 
 Given a set $\cal S $ of \emph{secured physical devices} of $E$, we write $\confCPSS E {\cal S} P$ to denote the resulting \CPS{}, and use 
$M$ and $N$ to range over \CPS{s}. We write $\confCPS E P$ when $\cal S = \emptyset$. 
\end{definition} 
In a \CPS{} $\confCPSS E {\cal S} P$, 
the ``secured'' devices in $\cal S$ are accessed in a protected way 
and hence they cannot be  attacked.\footnote{The presence  of battery-powered devices interconnected through wireless networks prevents the en-/decryption of all packets due to energy constraints.}

Let us now define physical environments $E$ and cyber components $P$ 
in order to formalise our proposal for modelling (and reasoning about)
\CPS{s} and cyber-physical attacks. 
\begin{definition}[Physical environment]
\label{def:physical-env}
Let $\hat{\mathcal{X}} \subseteq \mathcal{X}$ be a set of state variables,
$\hat{\mathcal{A}} \subseteq \mathcal{A}$ be a set of actuators, and
$\hat{\mathcal{S}} \subseteq \mathcal{S}$ be a set of sensors. A
\emph{physical environment} $E$ is an 8-tuple
$\envCPSS
{\statefun{}} 
{\actuatorfun{}} 
{\uncertaintyfun{}}  
{\evolmap{}}
{\errorfun{}}  
{\measmap{}}   
{\invariantfun{}}
{\safefun{}}
$,
where:
\begin{itemize}
\item $\statefun{} \in \mathbb{R}^{\hat{\cal X}} $ is the
\emph{state function},
\item $\actuatorfun{} \in \mathbb{R}^{\hat{\cal A}} $ is the
\emph{actuator function},
\item $\uncertaintyfun{} \in \mathbb{R}^{\hat{\cal X}} $ is the
\emph{uncertainty function},
\item $\evolmap{}: \mathbb{R}^{\hat{\cal X}} \times
\mathbb{R}^{\hat{\cal A}} \times \mathbb{R}^{\hat{\cal X}} \rightarrow
2^{\mathbb{R}^{\hat{\cal X}} }$ is the \emph{evolution map}, 
\item $\errorfun{} \in \mathbb{R}^{\hat{\cal S}}$ is the
\emph{sensor-error function},
\item $\measmap{}: \mathbb{R}^{\hat{\cal X}} \times
\mathbb{R}^{\hat{\cal S}} \rightarrow 2^{\mathbb{R}^{\hat{\cal S}} }$ is
the \emph{measurement map}, 
\item $\invariantfun{}: \mathbb{R} ^{\hat{\cal X}}  
\rightarrow \{\true, \false \}$ is the \emph{invariant function},
\item $\safefun{}: \mathbb{R} ^{\hat{\cal X}}  
\rightarrow \{\true, \false \}$ is the \emph{safety function}.
\end{itemize}
All the functions defining an environment are \emph{total functions}. 
\end{definition}

The \emph{state function} $\statefun{}$ returns the current value (in
$\mathbb{R}$) associated to each state variable of the system. The
\emph{actuator function} $\actuatorfun{}$ returns the current value
associated to each actuator. The \emph{uncertainty function}
$\uncertaintyfun{}$ returns the uncertainty associated to each state
variable. Thus, given a state variable $x \in \hat{\cal X}$,
$\uncertaintyfun{}(x)$ returns the maximum distance between the real value
of $x$ and its representation in the model. Later in the paper, we will be
interested in comparing the accuracy of two systems. Thus,  for
$\uncertaintyfun{}, \uncertaintyfun'{} \in \mathbb{R}^{\hat{\cal X}}$, we
will write $\uncertaintyfun{} \leq \uncertaintyfun'{}$ if
$\uncertaintyfun{}(x) \leq \uncertaintyfun'{}(x)$, for any $x \in
\hat{\cal X}$. Similarly, we write $\uncertaintyfun{} +
\uncertaintyfun'{}$ to denote the function $\uncertaintyfun''{} \in
\mathbb{R}^{\hat{\cal X}}$ such that $\uncertaintyfun''{}(x) =
\uncertaintyfun{}(x) + \uncertaintyfun'{}(x)$, for any $x \in \hat{\cal
X}$.

Given a state function, an actuator function, and an uncertainty function,
the \emph{evolution map} $\evolmap{}$ returns the set of next
\emph{admissible state functions}. It models the \emph{evolution law} of
the physical system, where changes made on actuators may reflect on state
variables. Since we assume an uncertainty in our models, $\evolmap{}$ does
not return a single state function but a set of possible state functions.
$\evolmap{}$ is obviously \emph{monotone} with respect to uncertainty: if
$\uncertaintyfun{} \leq \uncertaintyfun'{}$ then $\evolmap{}(\statefun{},
\actuatorfun{}, \uncertaintyfun{}) \subseteq \evolmap{}(\statefun{},
\actuatorfun{}, \uncertaintyfun'{})$. 

Both the state function and the actuator function are supposed to change during the evolution of the system, whereas the uncertainty function is 
constant.
Note that, although the
uncertainty function is constant, it can be used in the evolution map in
an arbitrary way (e.g., it could have a heavier weight when a state
variable reaches extreme values). Another possibility is 
to model the uncertainty function by means of a probability distribution.

The \emph{sensor-error function} $\errorfun{}$ returns the maximum error
associated to each sensor in $\hat{\cal S}$. Again due to the presence of
the sensor-error function, the \emph{measurement map} $\measmap{}$, given
the current state function, returns a set of admissible measurement
functions rather than a single one.

The \emph{invariant function} $\invariantfun{}$ represents the
conditions that the state variables must satisfy to allow for the
evolution of the system. A \CPS{} whose state variables don't satisfy the
invariant is in \emph{deadlock}.

The \emph{safety function} $\safefun{}$ represents the conditions that the state variables must satisfy to consider the \CPS{} in a safe state. Intuitively, if a \CPS{} gets in an unsafe state, then its functionality may get compromised.

In the following, we use a specific notation for the replacement of a single component of an environment with a new one of the same kind; for instance, for $E = \envCPSS
{\statefun{}} 
{\actuatorfun{}} 
{\uncertaintyfun{}}  
{\evolmap{}}
{\errorfun{}}  
{\measmap{}}   
{\invariantfun{}}
{\safefun{}}
$, we write $\replaceENV E {\uncertaintyfun{}}  {\uncertaintyfun'{}}$
to denote 
$\envCPSS
{\statefun{}} 
{\actuatorfun{}} 
{\uncertaintyfun'{}}  
{\evolmap{}}
{\errorfun{}}  
{\measmap{}}   
{\invariantfun{}}
{\safefun{}}
$.




Let us now introduce a \emph{running example\/} 
 to illustrate
our approach. We remark that while we
have kept the example simple, it is actually far from trivial and designed
to describe  a wide number of attacks.
 A more complex example (say, with $n$
sensors and $m$ actuators) wouldn't have been more instructive but just
made the paper more dense.
\begin{example}[Physical environment of the \CPS{} $\mathit{Sys}$]
\label{exa:sys-physical} 
Consider a \CPS{} $\mathit{Sys}$ in which the temperature of an engine is
maintained within a specific range by means of a cooling system. The
physical environment $\env$ 
 of $\mathit{Sys}$ is constituted by: (i) a state variable
$\mathit{temp}$ containing the current temperature of the engine, and 
a state variable
$\mathit{stress}$ keeping track of the level of stress of the mechanical parts of the engine due to high temperatures (exceeding $9.9$ degrees); this integer 
variable ranges from $0$, meaning no stress,  to $5$, for high stress;
(ii) an
actuator $\mathit{cool}$ to turn on/off the cooling system; (iii) a sensor
$s_{\mathrm{t}}$ (such as a thermometer or a thermocouple) measuring the temperature
of the engine; (iv) an uncertainty $\delta=0.4$ associated to the only
variable $\mathit{temp}$; (v) 
the evolution law for the 
two state variables: 
the variable $\mathit{temp}$ is increased 
(resp., is decreased)  of one degree per time
unit if the cooling system is inactive (resp., active), whereas
the variable  $\mathit{stress}$ contains an integer that is increased each time the current temperature is above 
$9.9$ degrees, and dropped to $0$ otherwise;
(vi) an error
$\epsilon =0.1$ associated to the only sensor $s_{\mathrm{t}}$; (vii) a measurement
map to get the values detected by sensor $s_{\mathrm{t}}$, up to its error
$\epsilon$; (viii) an invariant function saying that the system gets
faulty when the temperature of the engine gets out of the range $[0, 50]$; 
(ix) a safety function to say that the system moves to an 
unsafe state when the level of stress reaches the threshold $5$.


Formally,  $\env = \envCPS 
{\statefun{}} 
{\actuatorfun{}} 
{\uncertaintyfun{}}  
{\evolmap{}}
{\errorfun{}}  
{\measmap{}}   
{\invariantfun{}}$ with:
\begin{itemize}
\item $\statefun{} \in \mathbb{R} ^{\{\mathit{temp},\mathit{stress}\} }$ and 
$\statefun{}(\mathit{temp})=0$ and $\statefun{}(\mathit{stress}){=}0$;
\item $\actuatorfun{} \in \mathbb{R} ^{\{\mathit{cool}\} }$ and
$\actuatorfun{}(\mathit{cool})=\off$; for the sake of simplicity, we can
assume $\actuatorfun{}$ to be a mapping $\{ \mathit{cool} \} \rightarrow
\{ \on , \off\}$ such that $\actuatorfun{}(\mathit{cool})= \off$ if
$\actuatorfun{}(\mathit{cool}) \geq 0$, and $\actuatorfun{}(\mathit{cool})= \on$ if
$\actuatorfun{}(\mathit{cool}) < 0$;

\item $\uncertaintyfun{} {\in} \mathbb{R}^{\{\mathit{temp},\mathit{stress}\} }$, 
$\uncertaintyfun{}(\mathit{temp}){=}0.4{=}\delta$ and $\uncertaintyfun{}(\mathit{stress}){=}0$;

\item $\evolmap{}( \statefun^i{}, \actuatorfun^i{}, \uncertaintyfun{}) $
is the set of  $\xi \in \mathbb{R} ^{\{\mathit{temp},\mathit{stress}\} }$
such that:
\begin{itemize}
\item
$\xi(\mathit{temp}) = \statefun^i{}(\mathit{temp}) + 
\mathit{\mathit{heat}}(\actuatorfun^i{},\allowbreak \mathit{cool}) + \gamma $, 
with $ \gamma \in [- \delta, + \delta] $ and $\mathit{heat}(\actuatorfun^i{},\mathit{cool})=-1$ if
$\actuatorfun^i{}(\mathit{cool}) = \on$ (active cooling), and
$\mathit{heat}(\actuatorfun^i{},\mathit{cool})=+1$ if
$\actuatorfun^i{}(\mathit{cool}) = \off$ (inactive cooling);

\item $\xi(\mathit{stress}) = \min (5 \, , \,  \statefun^i{}(\mathit{stress}){+}1)$ if   $\statefun^i{}(\mathit{temp})>9.9$;  $\xi(\mathit{stress}) = 0$, otherwise;

\end{itemize} 

\item $\errorfun{} \in \mathbb{R}^{\{s_{\mathrm{t}} \}}$ and
$\errorfun(s_{\mathrm{t}})= 0.1=\epsilon$;

\item $\measmap{}(\statefun^i{}, \errorfun{})  {=} \big \{ \xi :
\xi(s_{\mathrm{t}}) {\in} [ \statefun^i{}(\mathit{temp}){-}\epsilon \, , \, \allowbreak
\statefun^i{}(\mathit{temp}){+} \epsilon ]  \big \}$;

\item $\invariantfun{}(\statefun{})=\true$ if $0 \leq \statefun{}(\mathit{temp})\leq 50$; \allowbreak $\invariantfun{}(\statefun{})=\false$, otherwise. 

\item $\safefun{}(\statefun{})=\true$ if $\statefun{}(\mathit{stress}) < 5$; \allowbreak $\safefun{}(\statefun{})=\false$, if $\statefun{}(\mathit{stress})\geq 5$ (the maximum value for $\mathit{stress}$ is $5$). 

\end{itemize}
\end{example}



Let us now formalise the cyber component of \CPS{s} in \cname{}.
Basically, we extend 
the \emph{timed process algebra TPL} of~\cite{HR95} with two main ingredients: 
\begin{itemize}
\item  two different constructs to read values detected at sensors and  write values on actuators, respectively; 
\item   special constructs to represent malicious activities on physical 
devices. 
\end{itemize}
The remaining constructs are the same as those of TPL. 
\begin{definition}[Processes]
\emph{Processes} are defined as follows: 
\begin{displaymath}
\begin{array}{rl}	
P,Q \Bdf & \nil \q\, \big| \q\, \tick.P \q\, \big| \q\, P \parallel Q \q\, 
\big| \q\, \timeout {\pi.P} {Q} 
\q\, \big|   \\[3pt]
& \ifelse b P Q \q\, \big| \q\, P{\setminus} c \q\, \big| \q\,  H \langle \tilde{w} \rangle .
\end{array}
\end{displaymath}
\end{definition}

We write $\nil$ for the \emph{terminated process}. The process $\tick.P$
sleeps for one time unit and then continues as $P$. We write $P \parallel
Q$ to denote the \emph{parallel composition}
 of concurrent \emph{threads}$P$ and $Q$.
 The process $\timeout {\pi.P} Q$, with $\pi\in
\{\OUT{c}{v},\LIN{c}{x}, \rsens x s, \wact v a, \allowbreak \rsens x
{\mbox{\Lightning}p}, \allowbreak \wact v {\mbox{\Lightning}p} \}$, denotes
\emph{prefixing with timeout}. Thus, $\timeout{\OUT c v . P}Q$ sends the
value $v$ on channel $c$ and, after that, it continues as $P$; otherwise,
after one time unit, it evolves into $Q$. The process $\timeout{\LIN c x.
P}Q$ is the obvious counterpart for reception.
The process 
 $\timeout{\rsens x s.P}{Q}$
reads the values detected by the sensor $s$, 
whereas $\timeout{\wact v a.P}{Q}$ writes 
on the  actuator $a$. 
For  $\pi \in \{ \rsens x {\mbox{\Lightning}p} , \wact v {\mbox{\Lightning}p} \}$, the process $\timeout{\pi.P}{Q}$ denotes the reading and the writing, respectively, of the physical device $p$ (sensor or actuator)
 made by the \emph{attack}. Thus, in \cname{}, \emph{attack processes} have specific constructs to interact with physical devices.

The process $P{\setminus}c$ is the channel restriction operator of CCS. It
is quantified over the set $\cal C$ of communication channels but we
sometimes write $P{\setminus}\{ c_1,
c_2, \ldots , c_n \}$ to mean
$P{\setminus}{c_1}{\setminus}{c_2}\cdots{\setminus}{c_n}$. 
The process $\ifelse b P Q$ is the standard conditional, where $b$ is a decidable guard. 
In processes of the form $\tick.Q$ and $\timeout {\pi.P} Q$, the occurrence of $Q$ is said to be \emph{time-guarded}. 
The process $H \langle \tilde{w} \rangle$ denotes (guarded) recursion. We  assume a set of \emph{process identifiers} ranged over by $H,H_1,H_2$.
We write $H \langle w_1,\ldots, w_k \rangle$ to denote a recursive process $H$ defined via an equation $H(x_1,\ldots, x_k) = P$, where (i) the tuple $x_1,\ldots, x_k$ contains all the variables that appear free in $P$, and (ii) $P$ contains only  guarded occurrences of the process identifiers, such as $H$ itself. We say that recursion is \emph{time-guarded} if $P$ contains only time-guarded occurrences of the process identifiers. Unless explicitly stated our recursive processes are always time-guarded.

In 
the two constructs $\timeout{\LIN c x. P}Q$ and $\timeout{ \rsens x \mu. P}Q$,
with $\mu \in \{ p, \mbox{\Lightning}p \}$, the variable $x$ is said to be
\emph{bound}. 
This gives rise to the standard notions of \emph{free/bound (process) variables} 
and \emph{$\alpha$-conversion}.
A term is \emph{closed} if it does not contain free 
 variables, and 
we assume to always work with closed processes: the absence of free variables is
preserved at run-time. As further notation, we write $T{\subst v x}$ for the substitution of all occurrences of the 
the free variable $x$ in $T$ with the value $v$.


Note that in \cname{}, a processes might use sensors and/or actuators which are 
not defined in the 
environment. To rule out ill-formed \CPS{s},
we   use the following  definition.
\begin{definition}[Well-formedness]
Given a process $P$ and an environment $E= \envCPSS 
{\statefun{}} 
{\actuatorfun{}} 
{\uncertaintyfun{}}  
{\evolmap{}}
{\errorfun{}}  
{\measmap{}}   
{\invariantfun{}}
{\safefun{}}
$, the \CPS{} $\confCPS E P$ is 
\emph{well-formed} if: (i) for any sensor $s$ mentioned in $P$, the function $\errorfun{}$ is defined on $s$; (ii) for any actuator $a$ mentioned in $P$, the function $\actuatorfun{}$ is defined on $a$. 
\end{definition}
Hereafter, we will always work with well-formed \CPS{s}.

Finally, we adopt some \emph{notational conventions}.  
To model \emph{time-persistent prefixing}, we write $\pi.P$ for 
the process defined via the  equation $\mathit{Rcv} = 
\timeout{\pi.P}{\mathit{Rcv}}$, where $\mathit{Rcv}$ does not occur in $P$. We write $\timeout{\pi}Q$ 
as an abbreviation for $\timeout{\pi.\nil}{Q}$, and 
$\timeout{\pi.P}{}$ as an abbreviation for $\timeout{\pi.P}{\nil}$. We write $\OUTCCS c$ (resp., $\LINCCS c$) when channel $c$ is used for pure synchronisation.
For $k\geq 0$, we write $\tick^{k}.P$ as a shorthand for $\tick.\tick. \ldots \tick.P$, where the prefix $\tick$ appears $k$ consecutive times. 
We write $\ifthen b P$ instead of $\ifelse b P \nil$. 
Let $M = \confCPSS E {\cal S} P$, we write 
$M\parallel Q$ for $\confCPSS E {\cal S} { (P\parallel Q) }$, and $M{\setminus}c$ for $\confCPSS E {\cal S} {P {\setminus}c}$. 

We can now finalise our running example. 
\begin{example}[Cyber component of the \CPS{} $\mathit{Sys}$]
\label{exa:sys}
Let us define the cyber component of the \CPS{} $\mathit{Sys}$ described in \autoref{exa:sys-physical}. 
We define two parallel processes: $\mathit{Ctrl}$ and $\mathit{IDS}$. 
The former models the \emph{controller} activity, consisting in reading the temperature sensor and in governing the cooling system via its actuator, whereas the latter models a simple \emph{intrusion detection system} that attempts to detect and signal abnormal behaviours of the system. Intuitively, $\mathit{Ctrl}$ senses the temperature of the engine at each time slot. When the 
\emph{sensed temperature} is above $10$ degrees, the controller activates the coolant. The cooling activity is maintained for $5$ consecutive time units. After that time, the controller synchronises with the $\mathit{IDS}$ component via a private channel $\mathit{sync}$, and then waits for \emph{instructions}, via a channel $\mathit{ins}$. The $\mathit{IDS}$ component checks whether the 
\emph{sensed temperature} is still 
above $10$. If this is the case, it sends an \emph{alarm} of ``high 
temperature'', via a specific channel, and then says to $\mathit{Ctrl}$ to keep cooling for other $5$ time units; otherwise, if the temperature is not above $10$, the $\mathit{IDS}$ component requires $\mathit{Ctrl}$ to stop the cooling activity.  
\begin{displaymath}
{\small  
\begin{array}{rcl}
\mathit{Ctrl} &  =  & \rsens x {s_{\operatorname{t}}} . \ifelse {x>10}
{ \mathit{Cooling} } { \tick.\mathit{Ctrl} } \\[2pt]
\mathit{Cooling}  & =  &   \wact{\on}{\emph{cool}}.\tick^5 . \mathit{Check}
\\[2pt]
\mathit{Check} & = & 
\OUTCCS{\mathit{sync}}. 
\LIN{\mathit{ins}}{y}.\mathsf{if} \, 
(y=\mathsf{keep\_cooling})\\
&& \, \{ \tick^5.\mathit{Check} \} \:
\mathsf{else} \;
\{  \wact{\off}{\mathit{cool}}.\tick .\mathit{Ctrl} \}\\[2pt]
\mathit{IDS} & = &  \LINCCS {\mathit{sync}}  .  \rsens x {s_{\operatorname{t}}} .  \mathsf{if} 
\, (x>10) \\
&& \,  \{ \OUT{\mathit{alarm}}{\mathsf{high\_temp}}. 
\OUT{\mathit{ins}}{\mathsf{keep\_cooling}}. \\
&& \, \tick.\mathit{IDS} \}  \; \mathsf{else} \; \{ \OUT{\mathit{ins}}{\mathsf{stop}}.\tick. \mathit{IDS} \}
\enspace . \end{array}
}
\end{displaymath}
Thus, the whole \CPS{} is defined as: 
\begin{displaymath}
\mathit{Sys} \; = \; \confCPS {\env} {(\mathit{Ctrl} \parallel \mathit{IDS})
{\setminus}\{ \mathit{sync}, \mathit{ins}\}} \,,
\end{displaymath}%
where $\env$ is the physical environment defined in \autoref{exa:sys-physical}.
We remark that, for the sake of simplicity, our $\mathit{IDS}$ component is quite basic: for instance, it does not check wether the temperature is too low. 
 However, it is straightforward to replace it with a more sophisticated one, containing more informative tests on sensor values and/or on actuators commands. 

\end{example}


\begin{figure}[t]
\centering
\includegraphics[width=5.5cm,keepaspectratio=true,angle=0]{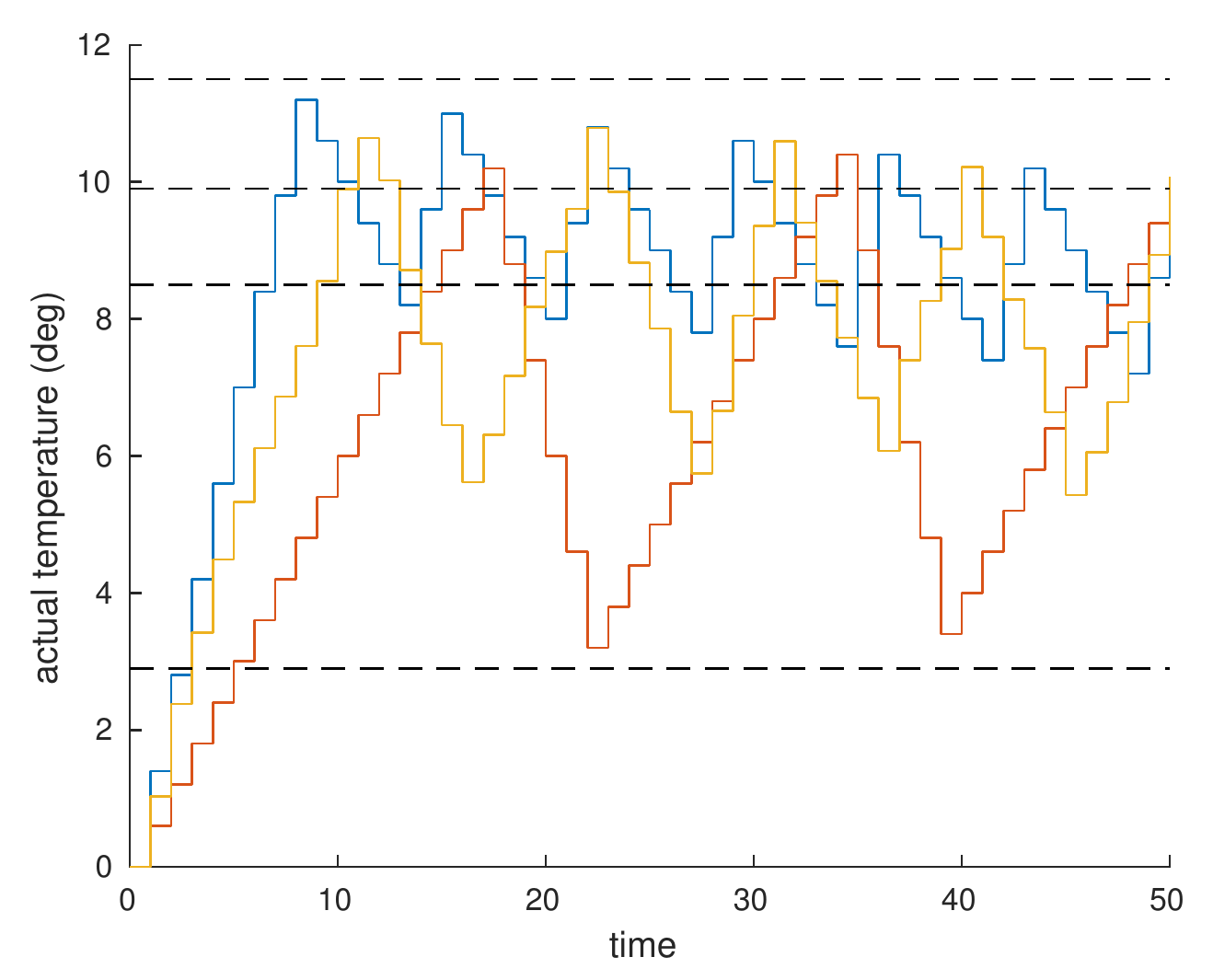}
\caption{Three  possible evolutions of the \CPS{} of \autoref{exa:sys}.}
\label{f:HS traj}
\end{figure}

\subsection{Labelled transition semantics}
\label{lab_sem}

\begin{table*}[t!]
\caption{LTS for processes}
\label{tab:lts_processes} 
\vspace*{-0.5cm}
\begin{displaymath}
\begin{array}{l@{\hspace*{10mm}}l}
\Txiom{Outp}
{-}
{ { \timeout{\OUT c v .P}Q } \trans{\out c v}   P}
&
\Txiom{Inpp}
{-}
{ { \timeout{\LIN c x .P}Q } \trans{\inp c v}    {P{\subst v x}}  }
\\[12pt]
\Txiom{Com}
{ P \trans{\out c v}  { P'}  \Q  Q \trans{\inp c v}  { Q'} }
{ P \parallel  Q \trans{\tau}  {P'\parallel Q'}}
&
\Txiom{Par}
{ P \trans{\lambda}  P' \Q \lambda \neq  \tick }
{ {P\parallel Q} \trans{\lambda} {P'\parallel Q}}
\\[12pt]
\Txiom{Write}
{ \mu \in \{ p , \mbox{\Lightning}p \} }  
{ { \timeout{\wact v \mu .P}Q } \trans{\snda \mu v}   P}
&
\Txiom{Read}
{  \mu \in \{ p , \mbox{\Lightning}p \} } 
{ { \timeout{\rsens x \mu .P}Q } \trans{\rcva \mu v}    {P{\subst v x}}  }
\\[12pt]
\Txiom{$\mbox{\Lightning}$SensWrite$\mbox{\,\Lightning}$}
{ P \trans{\snda {\mbox{\Lightning}s} v}  {P'}  \Q  Q \trans{\rcva s v}  { Q'} 
}
{ P \parallel  Q \trans{\tau:s}  {P'\parallel Q'}}
&
\Txiom{$\mbox{\Lightning}$ActRead$\mbox{\,\Lightning}$}
{ P \trans{\snda {a} v}  {P'}  \Q  Q \trans{\rcva {\mbox{\Lightning}a} v}  { Q'}
}
{ P \parallel  Q \trans{\tau:a}  {P'\parallel Q'}}
\\[12pt]
\Txiom{Res}{P \trans{\lambda} P' \Q \lambda \not\in \{ {\inp c v}, {\out c v} \}}{P {\setminus}c \trans{\lambda} {P'}{\setminus}c}
&
\Txiom{Rec}
{  P{\subst {\tilde{w}} {\tilde{x}}} \trans{\lambda}  Q \Q H(\tilde{x})=P}
{ H \langle \tilde{w} \rangle  \trans{\lambda}  Q}
\\[12pt] 
\Txiom{Then}{\bool{b}=\true \Q P \trans{\lambda} P'}
{\ifelse b P Q \trans{\lambda} P'}
&
\Txiom{Else}{\bool{b}=\false \Q Q \trans{\lambda} Q'}
{\ifelse b P Q \trans{\lambda} Q'}
\\[12pt]
\Txiom{TimeNil}{-}
{ \nil \trans{\tick}  \nil}
& 
\Txiom{Delay}
{-}
{  { \tick.P} \trans{\tick}  P}
\\[12pt]
\Txiom{Timeout}
{-}
{  {\timeout{\pi.P}{Q} }   \trans{\tick}  Q}
&
\Txiom{TimePar}
{
  P \trans{\tick}  {P'}  \Q 
   Q \trans{\tick} {Q'} 
}
{
  {P \parallel Q}   \trans{\tick}  { P' \parallel Q'}
}
\end{array}
\end{displaymath}
\end{table*}

In this section, we provide the dynamics of \cname{} in terms of a \emph{labelled transition system (LTS)} in the SOS style of Plotkin. 
\autoref{def:op-env} introduces  auxiliary
operators \nolinebreak on \nolinebreak environments.

\begin{definition}
\label{def:op-env}
Let $E = \envCPSS {\statefun{}} 
{\actuatorfun{}} 
{\uncertaintyfun{}}  
{\evolmap{}}
{\errorfun{}}  
{\measmap{}}   
{\invariantfun{}}
{\safefun{}}$. 
\begin{itemize} 
\item $\mathit{read\_sensor}(E,s)  \deff  \{\xi(s) \, : \,  \xi \in \measmap{}(\statefun{},\errorfun{})\}$,
\item $\mathit{update\_act}(E,a,v)  \deff
\replaceENV E {\actuatorfun{}}  {\actuatorfun{}[a \mapsto v]}$, 
\item $\mathit{next}(E)  \deff \bigcup_{\scriptsize \xi \in   \evolmap{}(\statefun{}, \actuatorfun{}, \uncertaintyfun{})} 
\{ \replaceENV E {\statefun{}}  {\xi}
\}$,
\item $\invariantfun{}(E)  \deff \invariantfun{}(\statefun{})$,
\item $\safefun{}(E)  \deff \safefun{}(\statefun{})$. 
\end{itemize}
\end{definition}

The operator 
$\mathit{read\_sensor}(E,s)$ returns the set of possible measurements detected by sensor $s$ in the environment $E$; it returns a set of possible values rather than a single value due to the error $\errorfun{}(s)$ of sensor $s$. 
$\mathit{update\_act}(E,a,v)$ returns the new environment in which the
actuator function is updated in such a manner to associate the actuator
$a$ with the value $v$.
$\mathit{next}(E)$ returns the set of the next
admissible environments reachable from $E$, by an application of $\evolmap{}$. $\invariantfun{}(E)$ checks whether the invariant is
satisfied by the current values of the state variables (here, abusing notation, we overload the meaning of the function
$\invariantfun{}$). 
$\safefun{}(E)$ checks whether the safety conditions are
satisfied by the current values of the state variables.

\begin{table*}[t]
\caption{LTS for \CPS{s}}
\label{tab:lts_systems} 
\begin{displaymath}
{\small 
\begin{array}{c}
\Txiom{Out}
{P \trans{\out c v}  P' \Q \operatorname{inv}(E)}
{\confCPSS E {\cal S} P   \trans{\out c v}   \confCPSS E {\cal S} {P' }}
\Q\Q\Q\Q
\Txiom{Inp}
{P  \trans{\inp c v}  P'\Q \operatorname{inv}(E)}
{\confCPSS E {\cal S} P    \trans{\inp c v}  \confCPSS E {\cal S} {P' }}
\\[13pt]
\Txiom{SensReadSec}{P \trans{\rcva s v} P' \Q s \in {\cal S} \Q \operatorname{inv}(E)\Q
\mbox{\small{$v \in \operatorname{read\_sensor}(E,s)$}} 
}
{\confCPSS E {\cal S} P \trans{\tau} \confCPSS E {\cal S} {P'}}
\\[13pt]
\Txiom{SensReadUnsec}{P \trans{\rcva s v} P' \Q s \not\in {\cal S} \Q P 
\ntrans{\snda{\mbox{\Lightning}s}{v}}\Q \operatorname{inv}(E)\Q
\mbox{\small{$v \in \operatorname{read\_sensor}(E,s)$}} 
}
{\confCPSS E {\cal S} P \trans{\tau} \confCPSS E {\cal S} {P'}}
\\[13pt]
\Txiom{$\mbox{\Lightning}$SensRead$\mbox{\,\Lightning}$}{P \trans{\rcva {\mbox{\Lightning}s} v} P'  \Q s \not\in {\cal S} \Q \operatorname{inv}(E)\Q
\mbox{\small{$v \in \operatorname{read\_sensor}(E,s)$}} 
}
{\confCPSS E {\cal S} P \trans{\tau} \confCPSS E {\cal S} {P'}}
\\[13pt]
\Txiom{ActWriteSec}{P \trans{\snda a v} {P'} \Q a \in {\cal S} \Q   \operatorname{inv}(E) \Q
{E'}=\operatorname{update\_act}(E,a,v)}
{\confCPSS E {\cal S} P \trans{\tau} \confCPSS {E'}{\cal S}{P'}}
\\[13pt]
\Txiom{ActWriteUnsec}{P \trans{\snda a v} {P'} \Q a \not\in {\cal S} \Q   P 
\ntrans{\rcva{\mbox{\Lightning}a}{v}} \Q \operatorname{inv}(E) \Q
{E'}=\operatorname{update\_act}(E,a,v)}
{\confCPSS E {\cal S} P \trans{\tau} \confCPSS {E'}{\cal S}{P'}}
\\[13pt]
\Txiom{$\mbox{\Lightning}$ActWrite$\mbox{\,\Lightning}$}{P \trans{\snda {\mbox{\Lightning}a} v} {P'}  \Q a \not\in {\cal S} \Q \operatorname{inv}(E) \Q
{E'}=\operatorname{update\_act}(E,a,v)}
{\confCPSS E {\cal S} P \trans{\tau} \confCPSS {E'}{\cal S}{P'}}
\\[13pt]
\Txiom{Tau}{(P \trans{\tau} P') \vee (P \trans{\tau:p} P'\q p \not\in {\cal S} )\Q \operatorname{inv}(E)}
{ \confCPSS E {\cal S} P \trans{\tau} \confCPSS E {\cal S} {P'}}
\Q\Q
\Txiom{Deadlock}
{\neg\operatorname{inv}(E)}
{ \confCPSS E {\cal S} P \trans{\dead} \confCPSS E {\cal S} {P}}
\\[13pt]
\Txiom{Time}{ P \trans{\tick} {P'} \Q 
\confCPSS E {\cal S} P \ntrans{\tau} \Q
\operatorname{inv}(E) \Q E' \in \operatorname{next}(E)  }
{\confCPSS E {\cal S} P \trans{\tick} \confCPSS {E'} {\cal S} {P'}}
\Q\Q
\Txiom{Safety}
{\neg\operatorname{safe}(E) \Q \operatorname{inv}(E)}
{ \confCPSS E {\cal S} P \trans{\unsafe} \confCPSS E {\cal S} {P}}
\end{array}
}
\end{displaymath}
\end{table*}

In \autoref{tab:lts_processes}, we provide transition rules for processes.
Here, the meta-variable $\lambda$ ranges over labels in the set 
$\{\tick, \tau, {\out c v}, {\inp c v}, \allowbreak \snda a v,\rcva s v, \snda {\mbox{\Lightning}p} v, \rcva {\mbox{\Lightning}p} v, \tau{:}p\}$. Rules \rulename{Outp}, \rulename{Inpp} and \rulename{Com} serve to model channel communication, on some channel $c$. Rules \rulename{Write} and
\rulename{Read} denote the writing/reading of some data on the physical
device $p$. Rule
\rulename{$\mbox{\Lightning}$SensWrite$\mbox{\,\Lightning}$} models an
\emph{integrity attack on sensor $s$}, where the controller of $s$ is
supplied with a fake value $v$ provided by the attack.
Rule \rulename{$\mbox{\Lightning}$ActRead$\mbox{\,\Lightning}$} models a
\emph{DoS attack to the actuator $a$}, where the update request of the controller is intercepted by the attacker and it never reaches the actuator. Rule \rulename{Par} propagates untimed actions over parallel components.
Rules \rulename{Res},  \rulename{Rec}, \rulename{Then} and \rulename{Else} 
 are standard. The following four rules model the passage of  time. The symmetric counterparts of rules \rulename{Com} 
and \nolinebreak \rulename{Par} \nolinebreak are \nolinebreak  omitted.

In \autoref{tab:lts_systems}, we lift the transition rules from processes
to systems. Except for rule \rulename{Deadlock}, all rules have a common
premise $\invariantfun{}(E)$: a system can evolve only if the invariant is
satisfied.
Here, actions, ranged over by $\alpha$, are in the set $\{\tau,
{\out c v}, {\inp c v}, \tick , \dead , \unsafe\}$. These actions denote:
internal activities ($\tau$); logical activities, more
precisely, channel transmission (${\out c v}$ and ${\inp c v}$); the
passage of time ($\tick$); and two specific  physical events: system
deadlock ($\dead$) and the violation of the safety conditions ($\unsafe$).
Rules \rulename{Out} and
\rulename{Inp} model transmission and reception, with an external system,
on a channel $c$. Rule \rulename{SensReadSec} models the reading of the
current data detected at a \emph{secured sensor} $s$, whereas rule
\rulename{SensReadUnsec} models the reading of an \emph{unsecured sensor}
$s$. In this case, since the sensor is not secured, the presence of a
malicious action $\snda {\mbox{\Lightning}s} w$ prevents the reading of
the sensor. We already said that rule
\rulename{$\mbox{\Lightning}$SensWrite$\mbox{\,\Lightning}$} of
\autoref{tab:lts_processes} models integrity attacks on an unsecured sensor
$s$, however, together with rule \rulename{SensReadUnsec}, it also serves
to model \emph{DoS attacks on an unsecured sensor $s$}, as the controller
of $s$ cannot read its correct value if the attacker is currently
supplying a fake value for it.

Rule \rulename{$\mbox{\Lightning}$SensRead$\mbox{\,\Lightning}$} allows
the attacker to read the confidential value detected at an unsecured sensor
$s$. Rule~\rulename{ActWriteSec} models the writing of a value $v$ on a
\emph{secured actuator} \nolinebreak $a$, whereas rule~\rulename{ActWriteUnsec} models the
writing on a \emph{unsecured actuator}  $a$. Again, if the actuator is
unsecured, the presence of an attack (capable of performing an action
$\rcva{\mbox{\Lightning}a}{v}$) prevents the correct access to the
actuator by the controller. Rule
\rulename{$\mbox{\Lightning}$ActWrite$\mbox{\,\Lightning}$} models an
\emph{integrity attack to an unsecured actuator $a$}, where the attack
updates the actuator \nolinebreak with \nolinebreak a \nolinebreak fake \nolinebreak value.

Note that our operational semantics ensures a preemptive power  to prefixes of the form $\wact v {\mbox{\Lightning}p}$ and $\rsens x {\mbox{\Lightning}p}$ \emph{on unsecured devices $p$}. This because an attack process  can always prevent the regular access to a unsecured physical device (sensor or actuator) by its controller. 
\begin{proposition}[Attack preemptiveness]
\label{prop:Attacker-preemptiveness}
Let  $M \allowbreak = \confCPSS E {\cal S} P$. 
\begin{itemize}[noitemsep]
\item If there is $Q$ such that $P \trans{\snda {\mbox{\Lightning}s} v} Q$,
with $s \not \in {\cal S}$,  then there is no $M'$ such that $M \trans{\tau} M'$ by an application of the rule \rulename{SensReadUnsec}. 
\item If there is $Q$ such that $P \trans{\rcva {\mbox{\Lightning}a} v} Q$, 
with $a \not \in {\cal S}$, then there is no $M'$ such that $M \trans{\tau} M'$ by an application of the rule \rulename{ActWriteUnsec}.
\end{itemize}
\end{proposition}

Rule \rulename{Tau} lifts non-observable actions from processes to
systems. This includes communications channels and attacks' accesses to
unsecured physical devices. A similar lifting occurs in rule
\rulename{Time} for timed actions, where $\operatorname{next}(E)$ returns
the set of possible environments for the next time slot. Thus, by an
application of rule \rulename{Time} a \CPS{} moves to the next physical
state, in the next time slot. Rule~\rulename{Deadlock} is introduced 
 to signal the violation of the invariant. When the invariant
is violated, a system deadlock occurs and then, in \cname{}, the system
emits a special action $\dead$, forever. Similarly, rule~\rulename{Safety}
is introduced to detect the violation of safety conditions. In this case,
the system may emit a special action $\unsafe$ and then continue its
evolution.


Now, having defined the actions that can be performed by a system, we can
easily concatenate these actions to define the possible execution traces
of the system. Formally, given a trace  $t = \alpha_1 \ldots
\alpha_n$, we will write $\trans{t}$ as an abbreviation for
$\trans{\alpha_1}\ldots \trans{\alpha_n}$. In the following, we
will use the function $\#\tick(t)$ to get the number of occurrences of the
 action $\tick$ \nolinebreak in  \nolinebreak $t$.

The notion of trace allows us to provide a formal definition of soundness for \CPS{s}: a \CPS{} is said to be \emph{sound} if it never deadlocks and never violates the safety conditions.
\begin{definition}[System soundness]
Let $M$ be a well-formed \CPS{}. We say that $M$ is \emph{sound} if whenever 
$M \trans{t} M'$, for some  $t$, both actions  $\dead$ and  $\unsafe$
never occur  in  $t$. 
\end{definition}
In our security analysis, we will focus on sound \CPS{s}. For instance, 
\autoref{prop:sys} says that our running example $\mathit{Sys}$ is 
sound 
and it never transmits on the channel $\mathit{alarm}$.
\begin{proposition} 
\label{prop:sys}
\label{prop:sys1}
\label{prop:sys2}
Let $\mathit{Sys}$ be the \CPS{} defined in \autoref{exa:sys}. 
If $\mathit{Sys} \trans{t} \mathit{Sys}'$, for some trace $t {=}\alpha_1 \ldots \alpha_n$, then $\alpha_i \in \{ \tau , \tick \}$, for any $i \in \{1, \ldots, n\}$. 
\end{proposition}

Actually, we can be quite precise on the temperature reached by
$\mathit{Sys}$ before and after the cooling:
 in each of the $5$
rounds of cooling, the temperature will drop of a value laying in the
real interval $[1{-}\delta, 1 {+} \delta]$, where $\delta$ is the uncertainty. 

\begin{proposition}
\label{prop:X}
Let $\mathit{Sys}$ be the \CPS{} defined in \autoref{exa:sys}. For any execution trace of $\mathit{Sys}$, we have:
\begin{itemize}
\item when $\mathit{Sys}$ \emph{turns on} the cooling, the value of
the state variable $\mathit{temp}$ ranges over $(9.9 \, , \,  11.5]$;
\item when $\mathit{Sys}$ \emph{turns off} the cooling, the value of
the 
variable $\mathit{temp}$ ranges over $(2.9, 8.5]$. 
\end{itemize}
\end{proposition}

\subsection{Behavioural semantics} 
We recall that the \emph{observable activities} in \cname{} are: time
passing, system deadlock, violation 
of safety conditions, and  channel communication.
Having defined a labelled transition semantics, we are ready to formalise 
our behavioural semantics, based on execution traces. 

We adopt a standard notation for weak transitions: we write $\Trans{}$ 
for $(\trans{\tau})^*$, whereas $\Trans{\alpha}$ means $\ttranst{\alpha}$, and finally $\ttrans{\hat{\alpha}}$ denotes $\Trans{}$ if $\alpha=\tau$ and $\ttrans{\alpha}$ otherwise. Given a trace $t = \alpha_1 {\ldots} \alpha_n$, we write 
$\trans{t}$
 for $\trans{\alpha_1}{\ldots} \trans{\alpha_n}$, and 
$\Trans{\hat{t}}$ as an abbreviation for $\Trans{\widehat{\alpha_1}} {\ldots} 
\Trans{\widehat{\alpha_n}}$.

\begin{definition}[Trace preorder]
\label{Trace-equivalence}
We write $M \sqsubseteq N$ if whenever $M \trans{t}
M'$, for some $t$, there is $N'$ such that $N\Trans{\hat{t}}N'$. 
\end{definition}
\begin{remark}
Unlike standard trace semantics, our trace preorder is able to observe deadlock thanks to the presence of the rule \rulename{Deadlock} and the special action $\dead$: if $M \sqsubseteq N$ and $M$ eventually deadlocks 
then also $N$ must eventually deadlock.
\end{remark}






As we are interested in examining timing aspects of attacks, such as
beginning and duration, we propose a timed variant of $\sqsubseteq$ up to
(a possibly infinite) time interval. Intuitively, we write $M
\sqsubseteq_{m..n} N$ if the \CPS{} $N$ simulates the execution traces of
$M$, except for the time interval $m..n$.

\begin{definition}[Trace preorder up to a time interval]
\label{Time-bounded-trace-equivalence}
 We write $M \sqsubseteq_{m..n} N$, 
for  $m \in \mathbb{N}^+$ and $ n \in \mathbb{N}^+ \cup
\infty$, with $m \leq n$,  if the following 
conditions hold: 
\begin{itemize}
\item $m$ is the minimum integer for which there is a trace $t$, with $\#\tick(t)=m-1$, such that $M \trans{t}$ and $N \not\!\!\Trans{\hat{t}}$;
%
%

\item $n$ is the infimum element of $\mathbb{N}^+ \cup \infty$, $n \geq m$, such that whenever $M \trans{t_1}M'$, with $\#\tick(t_1)=n-1$, there is $t_2$, with 
$\#\tick(t_1)=\#\tick(t_2)$, such that $N \trans{t_2}N'$, for some $N'$, and  $M' \sqsubseteq N'$.
\end{itemize}
\end{definition}
In the second item, note that 
 $\mathrm{inf}(\emptyset)=\infty$. 
Thus,  if $M \sqsubseteq_{m..\infty} N $  then $N$ simulates $M$ only in the first $m-1$ time slots. 

Finally, note that we could have equipped \cname{} with a
\emph{(bi)simulation-based} behavioural semantics rather than a
trace-based one, as done in~\cite{LaMe17} for a core of \cname{} with no
security features; however, our trace semantics is simpler than
(bi)simulation and it is sensitive to deadlocks of \CPS{s}. Thus, it is
fully adequate for the purposes of this paper.


\section{Cyber-Physical Attacks}
\label{sec:cyber-physical-attackers}

In this section, we use \cname{} to formalise a \emph{threat model} where
attacks can manipulate sensor and/or actuator signals in order to drive a
\emph{sound} \CPS{} into an undesired state~\cite{TeShSaJo2015}. 
 An attack may have
different levels of access to physical devices depending on the model
assumed. For example, it might be able to get read access to the sensors
but not write access; or it might get write-only access to the actuators
but not read-access. 
This level of
granularity is very important to model precisely how attacks can affect a
CPS~\cite{Cardenas2015}. For simplicity, in this paper we don't represent attacks on
communication channels as our focus is \nolinebreak on
\nolinebreak  attacks \nolinebreak to \nolinebreak physical 
\nolinebreak devices.

The syntax of our cyber-physical attack is a slight restriction of that
for processes: in terms of the form $\timeout{\pi.P}Q$, we require $\pi
\in \{ \wact v {\mbox{\Lightning}p}, \rsens x {\mbox{\Lightning}p} \}$.
Thus, we provide a syntactic way to distinguish attacks from genuine
processes.


\begin{definition}[Honest system]
A \CPS{} $\confCPSS E {\cal S} P$ is \emph{honest} if $P$ is honest, where 
$P$ is honest if it does not contain prefixes of the form $\wact v {\mbox{\Lightning}p}$ or $\rsens x {\mbox{\Lightning}p} $. 
\end{definition}


We group cyber-physical attacks in classes that describe both the malicious activity and the timing aspects of the attack.
Thus, let $\I$ be a set of malicious activities on a number of physical
devices, $m \in \mathbb{N}^+$ be the time slot when an attack starts, and
$n \in \mathbb{N}^+ \cup \infty$ be the time slot when the attack ends, we
say that an \emph{attack $A$ is of class $C \in [ \I \rightarrow {\cal
P}(m..n) ]$} if: (1) all possible malicious actions of $A$ coincide with
those contained in $\I$, (2) the first of those actions may occur in the
$m$-th time slot (i.e., after $m{-}1$ $\tick$-actions), and (3) the last
of those actions may occur in the $n$-th time slot (i.e., after $n{-}1$
$\tick$-actions). Actually, for $\iota \in \I$, $C(\iota)$ returns a (possibly
empty) set
of time instants when the attack tamper with the device $\iota$; this set
is contained in $m..n$. A class $C$ is always a total function. 
\begin{definition}[Class of attacks]
\label{def:attacker-class}	
Let ${\cal I} = \{ \mbox{\Lightning}p ? : p \in {\cal S} \cup
{\cal A} \} \cup \{ \mbox{\Lightning}p ! \, : p \in {\cal S} \cup {\cal A}
\}$ be the set of all possible \emph{malicious activities} on 
physical devices. Let 
 $m \in \mathbb{N}^{+}$,    $n \in \mathbb{N}^+ \cup
\infty$, with $m \leq n$. 
An attack $A$ is of \emph{class}
$C \in [ \I \rightarrow {\cal P}(m..n) ]$ whenever:
\begin{itemize}
\item 
{\small \begin{math}
C(\iota)=
\{  k  :  A \trans{t}\trans{\iota v} A' \,   \wedge \, 
k = \#\tick(t)+1  
\}
\end{math}},
 for  $\iota \in \I$; 
\item {\small 
\begin{math}
m = \inf \{ \, k \, : \, k \in C(\iota) 
  \, \wedge \, \iota \in \I \,  \}
\end{math};
}
\item {\small 
\begin{math}
n = \sup \{ \, k \, : \, k \in C(\iota) 
  \, \wedge \, \iota \in \I \,  \}. 
\end{math}
}
\end{itemize}
\end{definition}


According to the approach proposed 
in~\cite{FM99},
we can say that an attack $A$ affects a \emph{sound} \CPS{}  $M$ if the execution of the composed system $M \parallel A$ differs from that of the original system $M$, in an observable manner. Basically, a cyber-physical attack can influence the system under attack in at least two different ways:
\begin{itemize}
\item The system $M \parallel  A$ might deadlock when $M$ may not; this 
means that the attack $A$ 
affects the \emph{availability} of the system. We recall that 
in the context of \CPS{s}, deadlock is a particular severe 
physical event. 
\item The system $M \parallel  A$ might have non-genuine execution traces containing observables (violations of safety conditions or 
communications on channels) that cannot be reproduced by $M$; here the attack affects the \emph{integrity} of the system behaviour.
\end{itemize}

\begin{definition}[Attack tolerance/vulnerability]
\label{def:attack-tolerance}
Let $M$ be an honest and sound  \CPS{}. We say that $M$ 
is \emph{tolerant to an attack $A$} if 
$M \parallel A \, \sqsubseteq \,  M$. 
We say $M$ is \emph{vulnerable to
 an attack $A$} if there is a 
time interval $m..n$, with $m\in \mathbb{N}^+$ and 
 $n \in \mathbb{N}^+ \cup \infty$, such that $M \parallel A \: 
\sqsubseteq_{m..n} \,  M$. 
\end{definition}

%

Thus, if a system $M$ is vulnerable to an attack $A$ of class $C \in [\I
\rightarrow {\cal P}(m..n)]$, during the time interval $m'..n'$, then the
attack operates during the interval $m..n$ but it influences the system
under attack in the time interval $m'..n'$ (obviously, $m' \geq m$).
If $n'$ is finite we have a \emph{temporary attack}, otherwise we have a
\emph{permanent attack}. 
Furthermore, 
if $m'-n$ is big enough and
$n-m$ is small, then we have a quick nasty attack that affects the system
late enough to allow \emph{attack camouflages}~\cite{GGIKLW2015}. On the
other hand, if $m'$ is significantly smaller than $n$, then the attack
affects the observable behaviour of the system well before its termination
and the \CPS{} has good chances of undertaking countermeasures to stop the
attack. Finally, if $M \parallel A \trans{t}\trans{\dead}$, for some 
trace $t$, the  attack $A$ is called \emph{lethal}, as it 
is capable to
halt (deadlock) the \CPS{} $M$. This is obviously a permanent attack. 

Note that, according to \autoref{def:attack-tolerance}, the tolerance (or
vulnerability) of a \CPS{} also depends on the capability of the
$\mathit{IDS}$ component to detect and signal undesired physical
behaviours. In fact, the $\mathit{IDS}$ component might be designed to
detect abnormal physical behaviours going well further than deadlocks and
violations of safety conditions. 

In the following, we say that an attack is \emph{stealthy} if it is able
to drive the \CPS{} under attack into an incorrect physical state (either
deadlock or violation of the safety conditions) without being noticed by
the $\mathit{IDS}$ component.


In the rest of this section, we present a number of different attacks to
the \CPS{} $\mathit{Sys}$ described in \autoref{exa:sys}.

\begin{example}
\label{exa:att:DoS}
Consider the following \emph{DoS/Integrity attack}
on the (controller of) the actuator $\mathit{cool}$, of class $C \in [\I
\rightarrow {\cal P}(m..m)]$ with $C(\mbox{\Lightning}cool?)=C(\mbox{\Lightning}cool!)=\{ m \} $ and $C(\iota) = \emptyset$, for $\iota \not \in \{
\mbox{\Lightning}cool? , \mbox{\Lightning}cool! \}$; we call the attack $A_m$:
\begin{displaymath}
\tick^{m{-}1}.  \mathsf{timeout} \lfloor {\rsens x {\mbox{\Lightning}cool}}.\mathsf{if}   (x{=}{\off})   \{ {\wact {\off}{\mbox{\Lightning}cool}} \} \rfloor  . 
\end{displaymath}
Here, the attack $A_m$ operates exclusively in the $m$-th time slot, when
it tries to steal the cooling command (on or off) coming from the
controller, and fabricates a fake command to turn off the cooling system.
In practice, if the controller sends a command to turn off the coolant,
nothing bad will happen as the attack will put the same message back. 
When the controller sends (in the $m$-th time slot) a command to turn the
cooling on, the attack will drop the command. 
We recall that the controller will turn on the cooling only if the 
sensed temperature is greater than $10$ (and hence $\mathit{temp}> 9.9$); 
this may happen only if $m > 8$. 
Since the command to turn the cooling on is never re-sent by $\mathit{Ctrl}$, the temperature will continue to rise, and after $4$ time units
the system will violate the safety conditions emitting an action $\unsafe$, while the $\mathit{IDS}$ component will start sending
alarms every $5$ time units, until the whole system deadlocks because the temperature reaches the threshold of $50$ degrees. 

\end{example}

\begin{proposition}
\label{prop:att:DoS}
Let $\mathit{Sys}$ be the \CPS{} defined in \autoref{exa:sys}, and $A_m$
be the attack defined in \autoref{exa:att:DoS}. Then, 
\begin{itemize}
\item $\mathit{Sys} \parallel  A_m \; \sqsubseteq \; \mathit{Sys}$, for $m \leq 8$, 
\item
 $\mathit{Sys} \parallel  A_m \;\: 
\sqsubseteq_{m{+}4..\infty} \; \mathit{Sys}$, for $m > 8$. 
\end{itemize}
\end{proposition}

In this case, the $\mathit{IDS}$ component of $\mathit{Sys}$ is effective
enough to detect the attack with only one time unit delay.

\begin{example}
\label{exa:att:DoS2}
Consider the following \emph{DoS/Integrity attack} to the (controller of) sensor $s_{\mathrm{t}}$, of class $C \in [\I
\rightarrow {\cal P}(2..\infty)]$ such that 
$C( \mbox{\Lightning}s_{\mathrm{t}}?) = \{ 2 \}$, 
$C(\mbox{\Lightning}s_{\mathrm{t}}!) = 2..\infty$ and $C(\iota) = \emptyset$, 
for $ \iota \not \in \{\mbox{\Lightning}s_{\mathrm{t}}!, \mbox{\Lightning}s_{\mathrm{t}}?\}$:
\begin{displaymath}
\begin{array}{rcl} A &  = & \tick.  
\mathsf{timeout}\lfloor \rsens x {\mbox{\Lightning}s_{\mathrm{t}} }.B \langle x 
\rangle    \rfloor{}  \\
B(y) & = &  \mathsf{timeout} \lfloor { \wact {y } {\mbox{\Lightning}s_{\mathrm{t}} }.\tick.B \langle y \rangle } \rfloor
{B \langle y \rangle } 
\end{array}
\end{displaymath}
Here, the attack $A$ does the following actions in sequence: (i) she
sleeps for one time unit, (ii) in the following time slot, she reads
the current temperature $v$ at sensor $s_{\mathrm{t}}$, and (iii) for the rest of
her life, she keeps sending the same temperature $v$ to the controller of
$s_{\mathrm{t}}$.
In the presence of this attack, the process $\mathit{Ctrl}$ never activates the $\mathit{Cooling}$ component (and, hence, nor the $\mathit{IDS}$ component, which is the only one which could send an alarm) as it will always detect a temperature below $10$.
Thus, the compound system $\mathit{Sys} \parallel A$ will
move to an unsafe state until the invariant will be violated and the
system will deadlock. Indeed, in the worst scenario, 
after $\lceil \frac{9.9}{1{+}\delta} \rceil =\lceil \frac{9.9}{1.4} \rceil=8$ 
$\tick$-actions (in the $9$-th time slot)
  the value 
of $\mathit{temp}$ will be above $9.9$, and after further $5$ $\tick$-actions
(in the $14$-th time slot)
the system will violate the safety conditions emitting an $\unsafe$ 
action.  After  $ = \lceil \frac{50}{1.4} \rceil=36$ $\tick$-actions, in the $37$-th time slot, the invariant may be broken because the state variable $\mathit{temp}$ may reach $50.4$ degrees, and the system will also emit a $\dead$ action. 
Thus, $\mathit{Sys} \parallel A \; \sqsubseteq _{14..\infty} \; \mathit{Sys}$. 
This is a \emph{lethal} attack, as it causes a 
shut down of the system. It is also a  \emph{stealthy attack}
as it remains unnoticed until the end. 

In this attack, the $\mathit{IDS}$ 
component is completely ineffective as the sensor used by the component
is compromised, and there is not way for the $\mathit{IDS}$ to 
understand whether the sensor is under attack. A more sophisticated $\mathit{IDS}$
might have a representation of the  plant to recognise abnormal evolutions of the sensed temperature. In such case, the $\mathit{IDS}$ might switch on a second sensor, hoping that this
one has not
been compromised yet. Another possibility for the designer of the 
\CPS{} is to secure the sensor. Although this is not always possible, as 
encryption/decryption of all packets depends on  energy constraints
of the device.
\end{example}

Our semantics ensures that secured devices cannot be attacked, as stated by the following proposition.
\begin{proposition}
\label{prop:critical2}
Let $M = \confCPSS E {\cal S} P$ be an honest  and sound \CPS{}.
Let $C \in [\I \rightarrow {\cal P}(m..n)]$, with 
$ \{ p :  C(\mbox{\Lightning}p?) \cup C(\mbox{\Lightning}p!)   \neq \emptyset \} \subseteq \cal S$. Then  $M \parallel A \sqsubseteq M$, for any attack $A$ of class $C$. 
\end{proposition}

 Now, let us examine a similar but less severe attack. 
\begin{example}
\label{exa:att:integrity}
Consider the following DoS/Integrity attack  to the controller of sensor $s_{\mathrm{t}}$, of class $C \in [\I 
\rightarrow {\cal P}(1..n)]$, for $n>0$, with $C(\mbox{\Lightning}s_{\mathrm{t}}!)=C(\mbox{\Lightning}s_{\mathrm{t}}?)=1..n$ and $C(\iota) = \emptyset$, for 
$\iota \not \in \{\mbox{\Lightning}s_{\mathrm{t}}!,\mbox{\Lightning}s_{\mathrm{t}}?\}$:  
\begin{displaymath}
\begin{array}{rcl}
A_n   &  =   &  \mathsf{timeout}\lfloor \rsens x {\mbox{\Lightning}s_{\mathrm{t}} }. 
\mathsf{timeout}\lfloor  \{ \wact {x{-}2} {\mbox{\Lightning}s_{\mathrm{t}} }. \\
&& \tick.A_{n-1 } \rfloor {A_{n-1 }} \rfloor {A_{n-1 }}
\end{array}
\end{displaymath}
with $ A_0  = \nil$. 
In this attack, for $n$ consecutive time slots, 
$A_n$ sends to the controller the current sensed temperature decreased by 
an offset $2$. The effect of this attack on the system depends on the
duration $n$ of the attack itself: 
\begin{itemize}
\item for $n \leq 8$, the attack is harmless as the variable $\mathit{temp}$ may not reach a (critical) temperature above $9.9$;
\item for $n=9$, the variable $\mathit{temp}$ might reach a temperature above $9.9$ in the $9$-th time slot, and the attack would delay the activation of the cooling system of one time slot; as a 
consequence, the system  might get into an unsafe state in the time 
slots $14$ and $15$, but no alarm will be fired. 
\item for $n \geq 10$, the system may get into an unsafe state in the time slot $14$ and in the following $n-7$ time slots; this is not a \emph{stealthy attack} as the $\mathit{IDS}$ will fire the alarm at most two time slots later (in the $16$-th time slot); this is a \emph{temporary attack} which ends in the time slot $n+7$.
\end{itemize}
\end{example}

\begin{proposition}
\label{prop:att:integrity}
Let $\mathit{Sys}$ be the \CPS{} defined in \autoref{exa:sys}, and $A_n$ be the attack defined in \autoref{exa:att:integrity}. Then:
\begin{itemize}
\item $\mathit{Sys} \parallel  A_n \, \sqsubseteq \, \mathit{Sys}$, for $n \leq 8$,
\item $\mathit{Sys} \parallel  A_n \; \sqsubseteq_{14..15} \; \mathit{Sys}$, for $n =9$, 
\item $\mathit{Sys} \parallel A_n \; \sqsubseteq_{14..n{+}7} \; \mathit{Sys}$, for $n 
\geq 10$.
\end{itemize}
\end{proposition}

\subsection{A technique for proving attack tolerance/vulnerability}
In this subsection, we provide sufficient criteria to prove attack
tolerance/vulnerability to attacks of an arbitrary class $C$. Actually, we
do more than that, we provide sufficient criteria to prove attack
tolerance/vulnerability to all attacks of a class $C'$ which is somehow
``weaker'' than a given class $C$.
\begin{definition} Let $C_1, C_2 \in [ \I \rightarrow {\cal P}(m..n) ]$ be 
two classes of attacks. We say that $C_1$ is \emph{weaker} than $C_2$, written 
$C_1 \preceq C_2$, if $C_1(\iota) \subseteq C_2(\iota)$, for any $\iota 
\in \I$.
\end{definition}

 The idea is to define a
notion of \emph{most powerful attack} (also called \emph{top attacker}) of a 
given class $C$,  such 
that, if a \CPS{} $M$ tolerates the most powerful attack of class
$C$ then it also tolerates \emph{any} attack  of 
class $C'$, with $C' \preceq C$. 
We will provide a similar condition for 
attack vulnerability: let  $M$ be a \CPS{}  
vulnerable to $\mathit{Top}(C)$ in the time interval $m_1..n_1$;
then, for any attack $A$ of class $C'$, with $C' \preceq C$, if $M$ is 
vulnerable to $A$ then it is so for   a smaller time interval $m_2..n_2 \subseteq m_1..n_1$.

Our notion of top attacker has two extra ingredients with respect to the
cyber-physical attacks seen up to now: (i) \emph{nondeterminism}, and (ii)
time-unguarded recursive processes. This extra power of the top attacker
is not a problem as we are looking for sufficient criteria.

For what concerns nondeterminism, we assume a generic procedure $\mathit{rnd}()$ that given an arbitrary 
set ${\cal Z}$ returns  an element of ${\cal Z}$  chosen in a nondeterministic manner.
This procedure allows us to express \emph{nondeterministic choice\/},  
  $P \oplus Q$,  as an abbreviation for the process 
$\ifelse {\mathit{rnd}(\{\true,\false\})}  P Q $. 
Thus, let  $\iota \in  \{ \mbox{\Lightning}p ? : p \in {\cal S} \cup
{\cal A} \} \cup \{ \mbox{\Lightning}p ! \, : p \in {\cal S} \cup {\cal A}
\}$,  $m \in \mathbb{N}^{+}$, $n \in \mathbb{N}^{+}  \cup
\infty$, with $m \leq n$, and ${\cal T} \subseteq m..n$, 
we define the attack process $ \mathit{Att}( \iota , k, {\cal T})$  
 as the attack which may achieve the malicious activity $\iota$, at the time 
slot $k$, and which tries to do the same in all subsequent time slots
of ${\cal T}$. Formally, 
\begin{displaymath}
\begin{array}{l}
\mathit{Att}( \mbox{\Lightning}p?, k, {\cal T})  = \\ 
\Q \ifthen {k \in {\cal T}} {  
(\mathsf{timeout}\lfloor\rsens x {\mbox{\Lightning}p?}.\mathit{Att}( \mbox{\Lightning}p?, k, {\cal T}) \rfloor \\
 \Q\q {\mathit{Att}( \mbox{\Lightning}p?, k{+}1, {\cal T})})  \, \oplus \,   \tick.  \mathit{Att}( \mbox{\Lightning}p? , k{+}1, {\cal T})} \: \mathsf{else}\:
\\
 \Q\q  \ifelse { k < \mathrm{sup}({\cal T})}
{\tick. \mathit{Att}( \mbox{\Lightning}p?, k{+}1, {\cal T})}
{\nil}\\[2pt]
\mathit{Att}( \mbox{\Lightning}p!, k, {\cal T})  =\\ 
\Q  \ifthen {k \in {\cal T}} {
(\mathsf{timeout} \lfloor \wact {\mathit{rnd}(\mathbb{R})} {\mbox{\Lightning}p!}.\mathit{Att}( \mbox{\Lightning}p!, k, {\cal T}) \rfloor \\
 \Q \q
{\mathit{Att}( \mbox{\Lightning}p!, k{+}1, {\cal T})}) \, \oplus \,   \tick.  \mathit{Att}( \mbox{\Lightning}p! , k{+}1, {\cal T})} \: \mathsf{else} 
\\
 \Q\q \ifelse {k < \mathrm{sup}({\cal T})}
{\tick. \mathit{Att}( \mbox{\Lightning}p! , k{+}1, {\cal T})} {\nil}
\enspace . 
\end{array}
\end{displaymath}
Note that for ${\cal T} = \emptyset$ we assume 
$\mathrm{sup}({\cal T})=- \infty$. 

We can now use the definition above to formalise the notion
of most powerful attack of a given  class $C$. 
\begin{definition}[Top attacker]
Let $C \in [ \I \rightarrow {\cal P}(m..n)]$ be a class of attacks. We define 
\begin{center}
\(
 \mathit{Top}(C)  \; = \;  \prod_{\iota \in \I } \mathit{Att}( 
\iota , 1 , C(\iota))
\)
\end{center}
as the most powerful attack, or \emph{top attacker\/},  of class $C$. 
\end{definition}

The following result provides  soundness criteria for attack 
tolerance and attack vulnerability.
\begin{theorem}[Soundness criteria]
\label{thm:sound}
Let $M$ be an honest and sound  \CPS{}, $C$ an arbitrary class
of attacks, and 
$A$  an  attack of a class $C'$, with $C' \preceq C$.
\begin{itemize}
\item If 
 $M \parallel \mathit{Top}(C)  \, \sqsubseteq \, M$ then  
 $M \parallel A \, \sqsubseteq  \, M$. 
\item If 
 $M \parallel \mathit{Top}(C)   \sqsubseteq_{{m_1}..{n_1}}   M$  then  
 either $M \parallel A     \sqsubseteq   M$ or
 $M \parallel A    \sqsubseteq_{{m_2}..{n_2}}   M$, 
with  $m_2..n_2 \subseteq m_1 .. n_1$. 
\end{itemize}
\end{theorem}

\begin{corollary}
Let $M$ be an honest and sound \CPS{}, and $C$ a class of attacks. If
$\mathit{Top}(C)$ is not lethal for $M$ then any attack $A$ of class $C'$,
with $C' \preceq C$, is not a lethal attack for $M$. If $\mathit{Top}(C)$
is not a permanent attack for $M$ then any attack $A$ of class $C'$, with
$C' \preceq C$, is not a permanent attack for $M$.
\end{corollary}
 

\section{Impact of an attack}  
\label{sec:impact}

In the previous section,
we have grouped cyber-physical attacks by
focussing on the physical devices under attack and the timing aspects of
the attack (\autoref{def:attacker-class}). 
Then, we have provided a formalisation of when a \CPS{} should be considered tolerant/vulnerable to an attack 
(\autoref{def:attack-tolerance}).
%
In this section,
we show that it is important not only to demonstrate the tolerance
(or vulnerability) of a CPS with respect to certain attacks, but also
to evaluate the disruptive impact of those attacks on the
\CPS{}~\cite{GeKiHa2015}. 

The goal of this section is twofold: to provide a \emph{metric} to
estimate the impact of a successful attack on a \CPS{}, and to investigate
possible quantifications of the chances for an attack of being successful
when attacking a \CPS{}.

As to the metric,
we focus on the ability that an attack may
have to drag a CPS out of the correct behaviour modelled by its evolution
map, with the given uncertainty. Recall 
that $\evolmap{}$
is 
\emph{monotone} with respect to the uncertainty.
Thus, an increase of the uncertainty 
may translate into a widening of the range of the possible behaviours of the \CPS{}.

In the following, for $M = \confCPSS E {\cal S} P$, we write $\replaceENV
M {\psi} {\psi'}$ to mean $\confCPSS {\replaceENV E {\psi} {\psi'}} {\cal
S} P$.
\begin{proposition}[Monotonicity] 
\label{prop:monotonicity}
Let $M = \confCPSS E {\cal S} P$ be an honest and sound \CPS{},
 and $\uncertaintyfun{}$ 
the uncertainty  of $E$. 
If $\uncertaintyfun{} \leq \uncertaintyfun'{}$ and $M \trans{t} M'$ then 
$\replaceENV M {\uncertaintyfun{}}  {\uncertaintyfun'{}} \trans{t} \replaceENV 
{M'} {\uncertaintyfun{}}  {\uncertaintyfun'{}}$.
\end{proposition}

However, a wider uncertainty in the model doesn't always
correspond to a widening of the possible behaviours of the \CPS{}. In
fact, this depends on the \emph{intrinsic tolerance} of a \CPS{} with
respect to changes in the uncertainty function.
\begin{definition}[System $\xi$-tolerance] 
An honest and sound \CPS{} $M = \confCPSS E {\cal S} P$, where
 $\uncertaintyfun{}$ is
the uncertainty  of $E$,   is \emph{$\xi$-tolerant}, for $\xi \in \mathbb{R}^{\hat{\cal X}}$ and $\xi \geq 0$, if
\begin{center}
\begin{math}
\xi \, = \, \sup \big\{ \xi'  : \,
 \replaceENV M {\uncertaintyfun{}}  {{\uncertaintyfun{}}{+}{\eta}}
\sqsubseteq  M, 
\text{ for any } 0 \leq \eta \leq  \xi'     
 \big\}.
\end{math}
\end{center}
\end{definition}

Intuitively, if a \CPS{} $M$ has been designed with a given uncertainty
$\uncertaintyfun{}$, but $M$ is actually $\xi$-tolerant, with $\xi > 0$, then the uncertainty $\uncertaintyfun{}$ is somehow underestimated: the real uncertainty of  $M$  is given by $\uncertaintyfun{} {+} \xi$.
This information is quite important when trying to estimate the impact of an attack on a \CPS{}. In fact, if a system  $M$ has been designed with a given uncertainty $\uncertaintyfun{}$, but $M$ is actually $\xi$-tolerant, with $\xi > 0$, then an attack has (at least) a ``room for maneuver'' $\xi$ to degrade the whole \CPS{} without being observed (and hence detected). 
Thus, in general, the tolerance 
 $\xi$ should be as small as possible.
 
Let $\mathit{Sys}$ be the \CPS{} of \autoref{exa:sys}. In the rest 
of the section, with an abuse of notation, we will write  
$\replaceENV {\mathit{Sys}}  {\delta}  {\gamma} $ to denote  $\mathit{Sys}$
where the uncertainty of the variable $\mathit{temp}$ is $\gamma$.
\begin{example}
\label{exa:toll}
The \CPS{} $\mathit{Sys}$ of \autoref{exa:sys} is $\frac{1}{20}$-tolerant. 
This because, 
\begin{math}
\sup \big\{ \xi'  :   \replaceENV {\mathit{Sys}} \delta {\delta {+} \eta}  \sqsubseteq  \mathit{Sys} , \text{ for } 0 \leq \eta \leq  \xi'   \big\}
\end{math} is equal to $\frac{1}{20}$. 
 Since $ \delta + \xi = \frac{8}{20} + \frac{1}{20}=\frac{9}{20}$, the
proof of this statement relies on 
the following proposition.
\end{example}

\begin{proposition}
\label{prop:toll}
Let $\mathit{Sys}$ be the \CPS{} of \autoref{exa:sys}. Then:
\begin{itemize}
\item 
$ \replaceENV {\mathit{Sys}} \delta  \gamma \, \sqsubseteq \, \mathit{Sys} $, for $\gamma \in (\frac{8}{20}, \frac{ 9}{20})$, 
\item 

$ \replaceENV  {\mathit{Sys}} \delta \gamma \, \not\sqsubseteq \,  \mathit{Sys}$, for $\gamma >\frac{ 9}{20}$.  

\end{itemize}
\end{proposition}

%

\autoref{f:Ex2}
shows an evolution of  $\replaceENV {\mathit{Sys}} {\delta} {\frac{29}{30}}$: the red box denotes a violation
of the safety conditions  
because the cooling cycle wasn't sufficient to
drop the (sensed) temperature below $10$ (here, the controller imposes $5$
further time units of cooling).

\begin{figure}[t]
\centering
\includegraphics[width=5.5cm,keepaspectratio=true,angle=0]{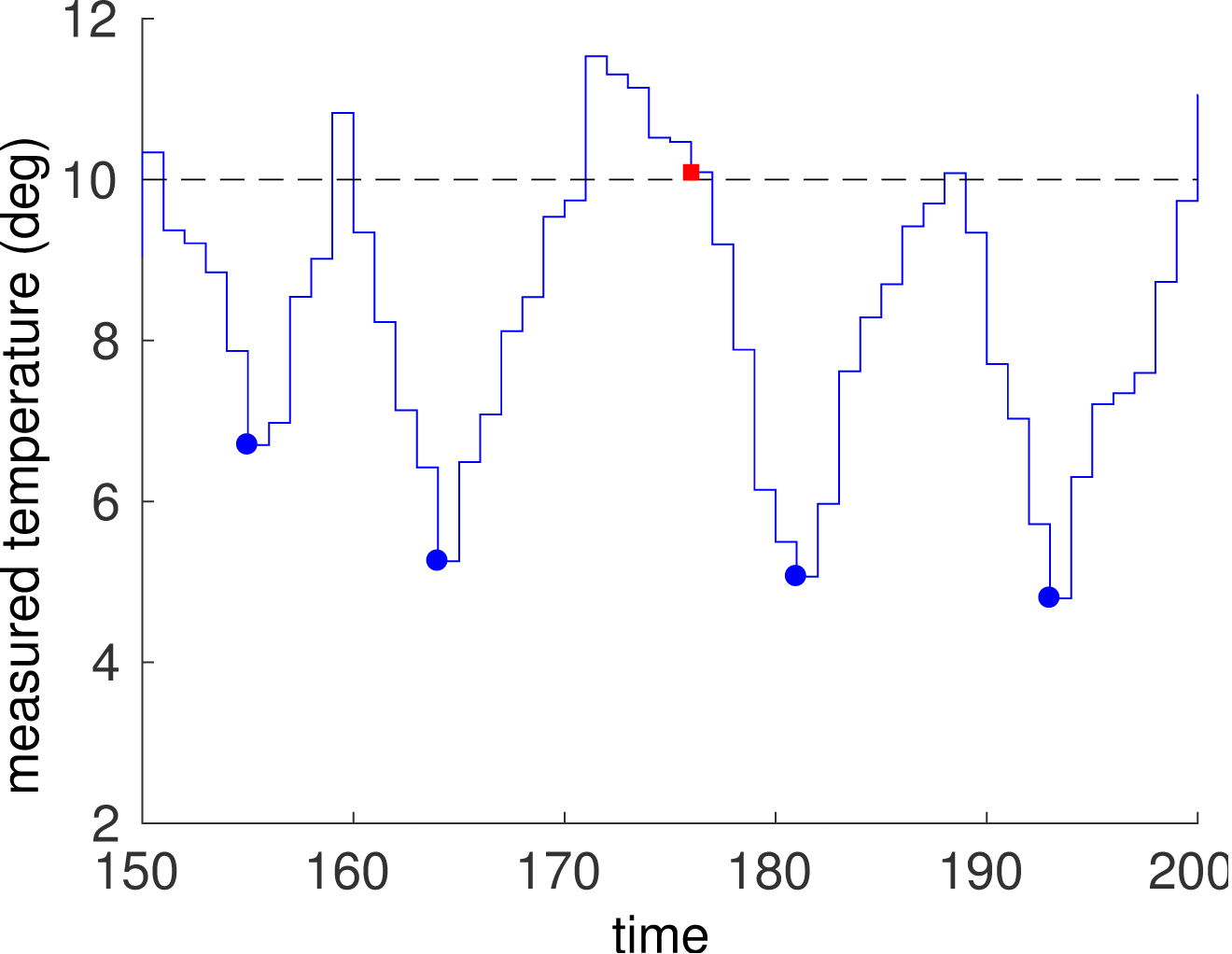}
\caption{Simulation of  $\replaceENV {\mathit{Sys}} {\delta} {\frac{19}{20}}$.}
\label{f:Ex2}
\end{figure}


Everything is in place to define our metric to estimate the 
impact of an attack.
\begin{definition}[Impact] 
\label{def:attack-xi-tolerance}
Let  $M= \confCPSS E {\cal S} P $ be an honest and sound \CPS{}, 
 where
 $\uncertaintyfun{}$ is
the uncertainty  of $E$. 
We say that an attack $A$ has  \emph{definitive impact} $\xi$ on the system $M$ if 
\begin{displaymath}
\xi  = \inf \big\{ \xi'  :  \xi' \in \mathbb{R}^{\hat{\cal X}} 
\: \wedge \: \xi'{>}0 \:  \wedge \: M \parallel  A  \sqsubseteq 
{\replaceENV M {\uncertaintyfun{}}  {{\uncertaintyfun{}}{+}{\xi'}}  }  \big\}.
\end{displaymath}%
It has \emph{pointwise impact} $\xi$ on the system 
$M$ at time $m$ if 
{\small 
\begin{displaymath}
\xi  {=}  \inf \big\{ \xi'  :  \xi' {\in} \mathbb{R}^{\hat{\cal X}} 
\,  \wedge \, M \parallel  A   \sqsubseteq_{m..n} 
\replaceENV M  {\uncertaintyfun{}}  {{\uncertaintyfun{}} {+} {\xi'}},  n {\in} \mathbb{N} {\cup} \infty   \big\}.
\end{displaymath}
}
\end{definition}

Intuitively, with this definition, we can establish either the definitive (and hence maximum) impact of the attack $A$ on the system $M$, or the impact at a specific time $m$.
In the latter case, by definition of $\sqsubseteq_{m..n}$, there are two possibilities: either the impact of the attack keeps growing after time $m$, or in the time interval $m{+}1$, the system under attack deadlocks.
 
The impact of $\mathit{Top}(C)$ provides an upper bound for the impact of all attacks of class $C'$, with $C' \preceq C$.  
\begin{theorem}[Top attacker's impact]
\label{thm:sound2}
Let  $M$ be an honest and sound \CPS{}, and 
 $C$  an arbitrary class of attacks. Let $A$ be an attack of class $C'$, with $C' \preceq C$.  
\begin{itemize}
\item The definitive impact of $\mathit{Top}(C)$ on $M$ is greater than or equal to the definitive impact of $A$ on $M$.
\item If $\mathit{Top}(C)$ has pointwise impact $\xi$ on $M$ at time $m$, and $A$ has pointwise impact $\xi'$  on $M$ at time $m'$, with  
$m' \leq m$, then $\xi' \leq \xi $.
\end{itemize}
\end{theorem}


\begin{example}
\label{exa:effect2}
Let us consider the attack $A$ of \autoref{exa:att:DoS2}. 
Then, $A$ has a definitive impact of $8.5$ on the \CPS{} $\mathit{Sys}$ defined in \autoref{exa:sys}. 
 Formally, 
\begin{math}
8.5= \inf \big \{ \, \xi'  :\: \xi'> 0 \: \wedge \: \mathit{Sys} \parallel  A \, \sqsubseteq \, \replaceENV {\mathit{Sys}} \delta  {\delta{+}\xi'} \big \}. 
\end{math}
Here, the attack can prevent the activation of the cooling system, 
and the  temperature will keep growing until the \CPS{} before enters continuously in an unsafe state and eventually deadlocks. 
Since $\delta + \xi = 0.4+8.5=8.9$, the proof of this statement relies on 
the following proposition.
\end{example}

\begin{proposition}
\label{prop:effect2}
 Let $\mathit{Sys}$ be the \CPS{} defined in \autoref{exa:sys}, and $A$ be the attack defined in \autoref{exa:effect2}. Then:
\begin{itemize}
\item  
$\mathit{Sys} \parallel A \, \not \sqsubseteq \, 
\replaceENV {\mathit{Sys}}
 \delta  \gamma$, for $\gamma \in (0.4,8.9)$, 
\item 
 $\mathit{Sys} \parallel A \, \sqsubseteq \, \replaceENV {\mathit{Sys}}
 \delta  \gamma$, for $\gamma >8.9$. 
\end{itemize} 
\end{proposition}

\autoref{def:attack-xi-tolerance} provided an instrument to estimate the
impact of a successful attack. However, there is at least another question
that 
a \CPS{} designer could ask: ``Is there a way to
estimate the chances that an attack will be successful during the
execution of my \CPS{}?'' To paraphrase in a more operational manner: how
many execution traces of my \CPS{} are prone to be attacked by a specific
attack?

For instance, consider again the simple attack $A_m$  proposed in
\autoref{exa:att:DoS}:
\begin{displaymath}
\tick^{m{-}1}.  
\timeout{\rsens x {\mbox{\Lightning}cool}.
\mathsf{if}  \, 
(x{=}{\off}) \,  { \{ \wact{\off}{\mbox{\Lightning}cool}} \}} {}.
\end{displaymath}
Here, in the $m$-th time slot the attack tries to eavesdrop a
command 
to turn on the cooling. The attack is very
quick and condensed in a single time slot. The question is: what are
the chances of success of such a quick attack?


\begin{figure}[t!]
\centering
\includegraphics[width=5.5cm,keepaspectratio=true,angle=0]{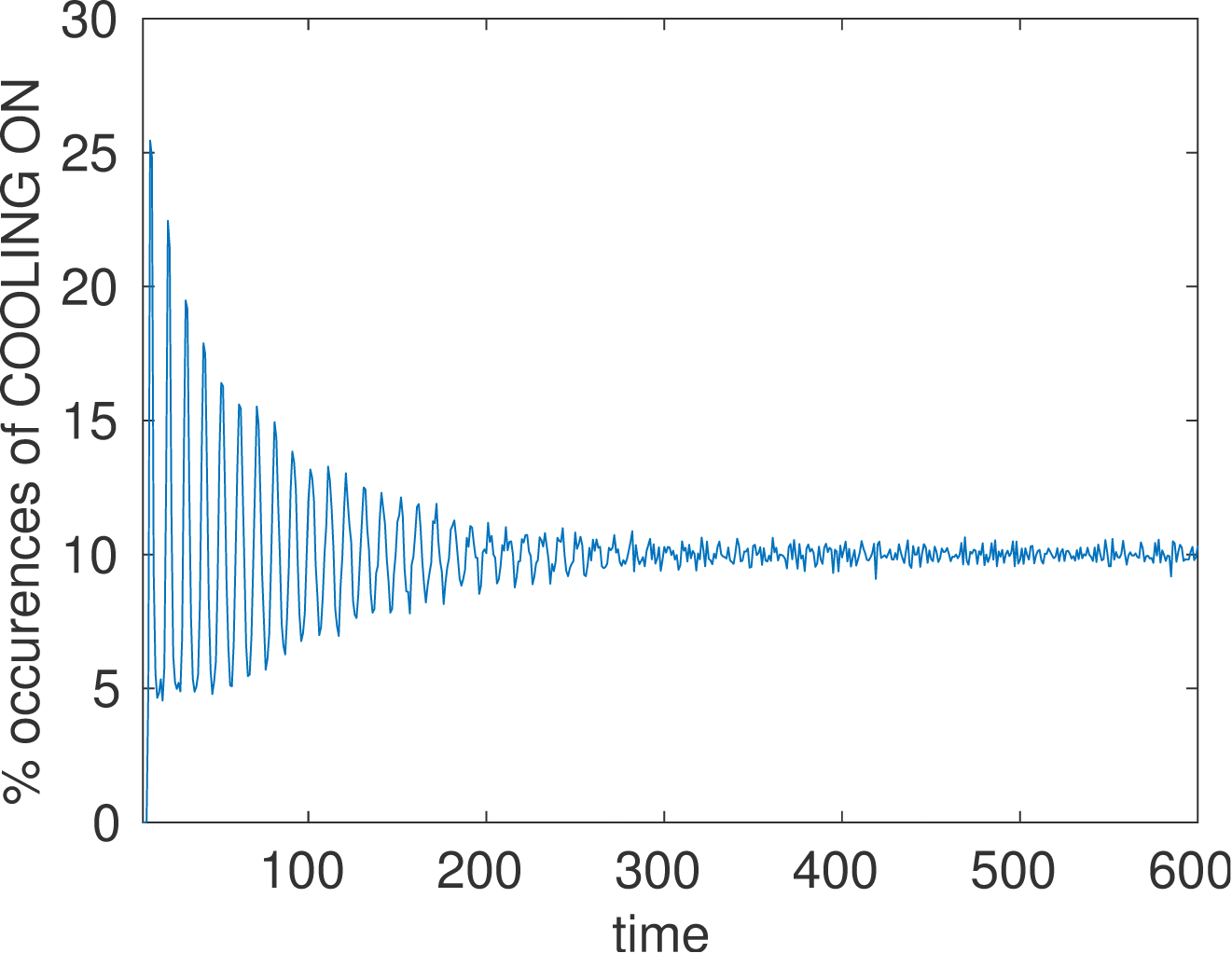}
\caption{A quantitative analysis of the attack of \autoref{exa:att:DoS}.}
\label{fig:cooling-trace}
\end{figure}

\autoref{fig:cooling-trace} provides a representation  of an
experiment in MATLAB where we launched $10000$ executions of our \CPS{} in
isolation, lasting $700$ time units each. From the aggregated data
contained in this graphic, we note that after a transitory phase (whose
length depends on several things: the uncertainty $\delta$, the initial
state of the system, the length of the cooling activity, etc.) that lasts
around $300$ time slots, the rate of success of the attack $A_m$ is
around $10\%$. 
The reader may wonder why exactly the $10\%$.
This depends on the periodicity of our \CPS{}, as in average the
cooling is activated every $10$ time slots. 

\begin{figure*}[t]
\centering
\includegraphics[width=5.5cm,keepaspectratio=true,angle=0]{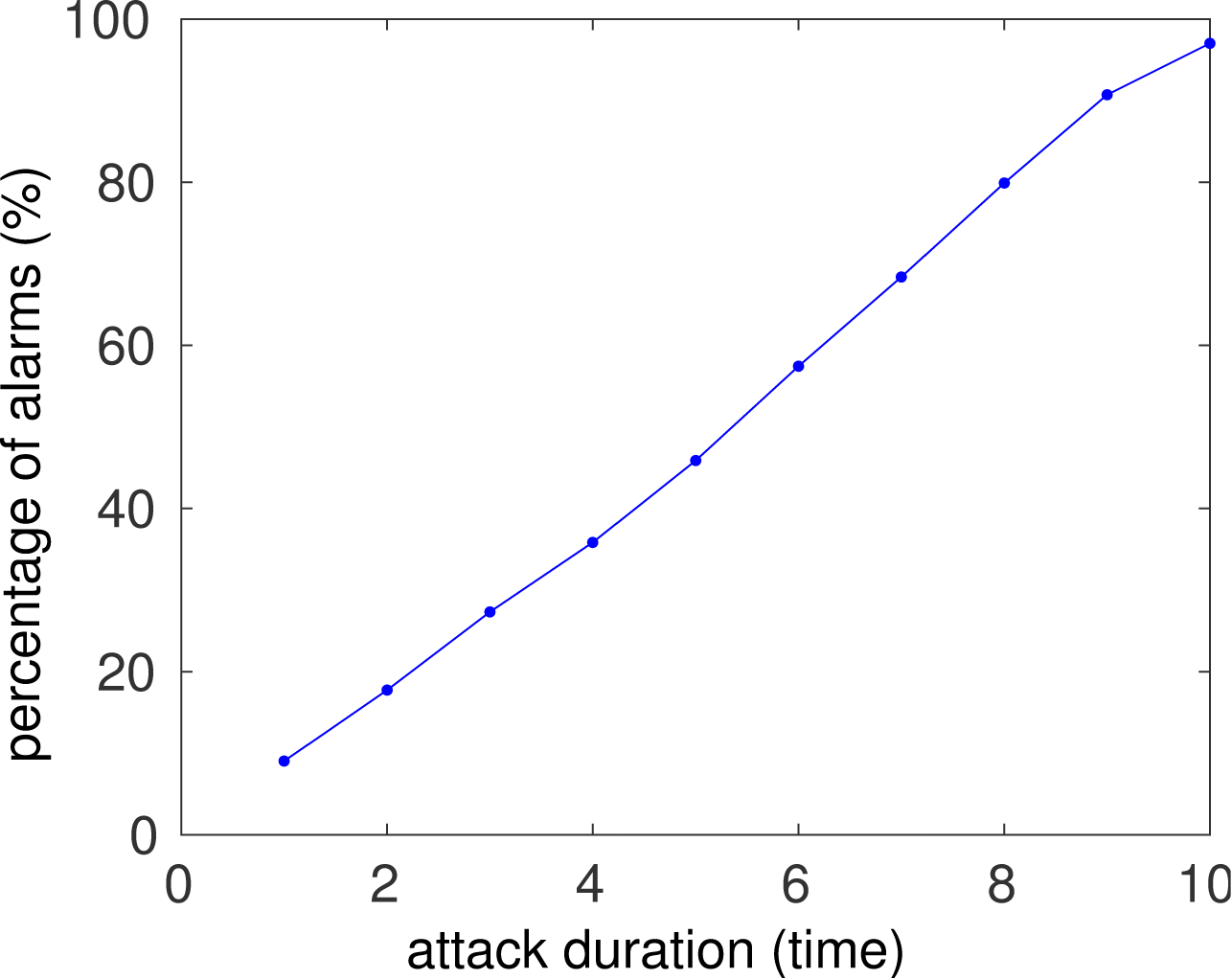}
 \Q\Q\Q\Q\Q\Q
\medskip
\includegraphics[width=5.5cm,keepaspectratio=true,angle=0]{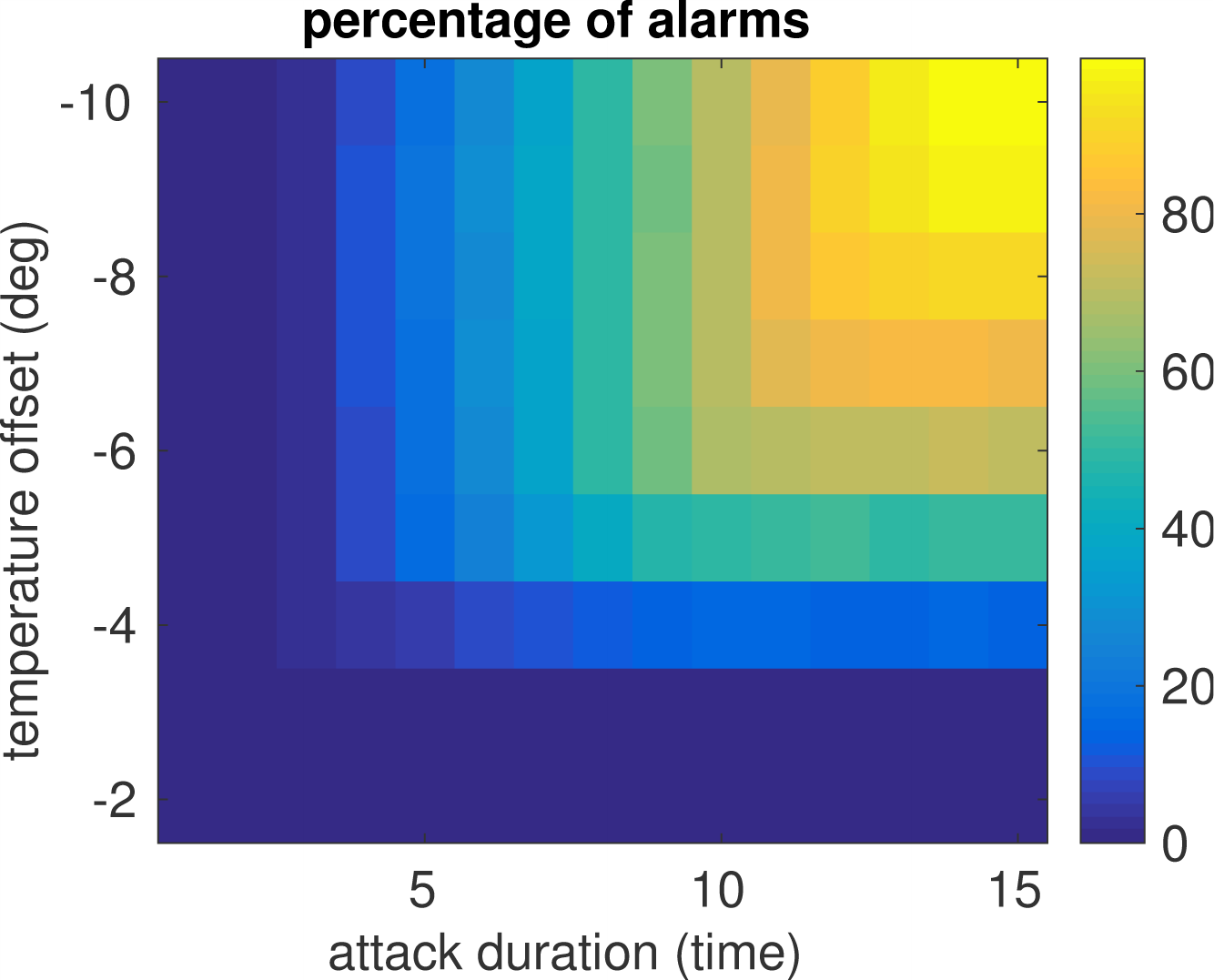}
\caption{A quantitative analysis of two different attacks.}
\label{fig:attack-trace}
\end{figure*}

This example shows that, as pointed out in~\cite{GGIKLW2015}, the
effectiveness of a cyber-physical attack depends on the information the
attack has about the functionality of the whole  \CPS{}.
%
%
For instance, if the attacker were not aware of the exact periodicity of
the \CPS{}, she might try, if possible, to repeat the attack on more
consecutive time slots. In this case, the left graphic of
\autoref{fig:attack-trace} says that the rate of success of the attack
increases linearly with the length of the attack itself (data obtained by attacking the \CPS{} after the transitory period). Thus,
if the attack of \autoref{exa:att:DoS} were iterated for $10$ 
time slots, say  
\begin{displaymath}
\begin{array}{rcl}
A^{10}_m & = & \tick^{m{-}1}.  B_{10}\\[1pt]
 B_i & = &  \mathsf{timeout} \lfloor
{\rsens x {\mbox{\Lightning}cool}}.\mathsf{if} \, (x={\off})\, \\
&& \{ {\wact {\off}{\mbox{\Lightning}cool}}.\tick.B_{i-1} \}
 \rfloor
{B_{i-1}},  \mbox{\small $\textrm{ for } 1 \leq i \leq 10$}
\end{array}
\end{displaymath}
with $B_0 = \nil$, 
the rate of success would be almost
$100\%$.

Finally, consider a generalisation of the attack of
\autoref{exa:att:integrity}:
\begin{displaymath}
\begin{array} {rcl}
A_0^k & =  & \nil \\[1pt]
A^k_n &  =  & \mathsf{timeout} \lfloor \rsens x {\mbox{\Lightning}s_{\mathrm{t}}}.
\mathsf{timeout} \lfloor  \wact {x{-}k} {\mbox{\Lightning}s_{\mathrm{t}} }.
\\
&& \tick.A_{n-1 }^k \rfloor  {A_{n-1}^k}
\rfloor {A_{n-1 }^k}
\end{array} 
\end{displaymath} 
for $1 \leq n \leq 15$ and $2 \leq k \leq 10$. Here, the attack decreases the sensed temperature of an offset $k$. Now, suppose to launch this attack after, say, $300$ time slots (i.e., after the transitory phase). Formally, we define the attack: 
$B^k_n \; = \; \tick^{300}.A_{n }^k$. 
In this case, the right graphic of \autoref{fig:attack-trace} provides a graphical representation of the percentage of alarms on $5000$ execution traces lasting $100$ time units each. Thus, for instance, an attack lasting $n=8$ time units with an offset $k=5$ affects around $40\%$ of the execution traces of the \CPS{}.


\section{Conclusions, related and future work}
\label{sec:conclusions}

We have provided formal \emph{theoretical foundations} to reason about and 
statically detect
attacks to physical devices of \CPS{s}.
To that end, we have
proposed a hybrid process calculus, called \cname{}, as a formal
\emph{specification language} 
to model physical and
cyber components of \CPS{s} as well as cyber-physical attacks.
Based on \cname{} and its labelled transition semantics, we have formalised a
threat model for \CPS{s} by grouping attacks in classes,
according to the target physical devices
and two timing parameters: begin and duration of the attacks. Then, we
relied on the trace semantics of \cname{} to assess
\emph{attack tolerance/vulnerability} with respect to a given attack. 
Along the lines of GNDC~\cite{FM99}, we defined a notion of \emph{top
attacker}, $\mathit{Top}(C)$, of a given class of attacks $C$, which has
been used to provide sufficient criteria to prove attack
tolerance/vulnerability to all attacks of class $C$ \nolinebreak (and
\nolinebreak weaker \nolinebreak ones).
Here, would like to mention that in our companion paper~\cite{LaMe17} we developed a \emph{bisimulation congruence} for a  simpler version of the calculus where security features have been completely stripped off.  For simplicity, in the current submission, we adopted as main behavioural equivalence trace equivalence instead of bisimulation. 
We could switch to a bisimulation semantics, preserved by parallel composition, which would  allow us to scale our verification method (\autoref{thm:sound}) to bigger systems. 
Finally, we have provided a metric to estimate the impact of a successful attack on a \CPS{} 
together with possible quantifications of the success chances of an attack.
We proved that the impact of the most powerful attack $\mathit{Top}(C)$ represents an upper bound for the impact of any attack $A$ of class $C$ \nolinebreak (and \nolinebreak weaker \nolinebreak ones).

 We have illustrated our concepts by means of a running example,
focusing in particular on a formal treatment of both integrity and DoS
attacks to sensors and actuators of \CPS{s}. Our example is simple but far from trivial and designed to describe a wide number of attacks. 

\subsubsection*{Related work}
Among the $118$ papers discussed in the comprehensive
survey~\cite{survey-CPS-security-2016}, $50$ adopt a discrete notion of
time similar to ours, $13$ a continuous one, $48$ a quasi-static time
model, and the rest use a hybrid time model. Most of these papers
investigate attacks on \CPS{s} and their protection by relying on
\emph{simulation test systems} to validate the results.

We focus on the papers that are most  related to our work.
 Huang et
al.~\cite{HCALTS2009} were among the first to propose \emph{threat models}
for \CPS{s}. Along with~\cite{KrCa2013,BestTime2014}, they stressed the
role played by timing parameters on integrity and DoS attacks. Alternative
threat models are discussed in~\cite{GeKiHa2015,GGIKLW2015,TeShSaJo2015}.
In particular, Gollmann et al.~\cite{GGIKLW2015} discussed possible goals
(\emph{equipment damage}, \emph{production damage}, \emph{compliance
violation}) and \emph{stages} (\emph{access}, \emph{discovery},
\emph{control}, \emph{damage}, \emph{cleanup}) of cyber-physical attacks.
In the current paper, we focused on the damage stage, where the attacker already has a
rough idea of the plant and the control architecture of the target \CPS{}.

A number of
works use
formal methods for \CPS{} security, although they apply methods, and most
of the time have goals, that are quite different from ours.

Burmester et al.~\cite{BuMaCh2012} employed \emph{hybrid timed automata} to give a threat 
framework based on the traditional Byzantine faults model for
crypto-security.
However, as remarked in~\cite{TeShSaJo2015},
cyber-physical attacks and faults have inherently distinct
characteristics.
Faults are considered as physical events that affect the system behaviour,
where simultaneous events 
don't act in a coordinated way; 
cyber-attacks may be performed over a significant number of attack points and in a coordinated way.

In~\cite{Vig2012}, Vigo presented an attack scenario that addresses some
of the peculiarities of a cyber-physical adversary, and discussed how this
scenario relates to other attack models popular in the security protocol
literature. Then, in~\cite{Vigo2015,VNN2013} Vigo et al.\ proposed an
untimed calculus of broadcasting processes 
equipped with 
notions of 
failed and unwanted communication. These works differ quite considerably from ours, e.g., they focus on DoS attacks without taking into consideration timing aspects or impact of the attack.

C\'ombita et al.~\cite{Cardenas2015} and Zhu and
Basar~\cite{game-theory-CPS2015} applied \emph{game theory} to capture the
conflict of goals between an attacker who seeks to maximise the damage
inflicted to a \CPS{}'s security and a defender who aims to minimise
it~\cite{game-theory-2013}.

Finally, there are three  recent papers that were developed in
parallel to ours: \cite{Nigam-Esorics2016,RocchettoTippenhauer2016a,RocchettoTippenhauer2016b}.  Rocchetto and
Tippenhaur~\cite{RocchettoTippenhauer2016a} introduced a taxonomy of the diverse attacker models proposed
for \CPS{} security and outline requirements for generalised attacker
models; in~\cite{RocchettoTippenhauer2016b}, they then proposed an
extended Dolev-Yao attacker model suitable for \CPS{s}. 
In their approach, physical layer interactions are modelled as abstract
interactions between logical components to support reasoning on the
physical-layer security of \CPS{s}. This is done by introducing additional
orthogonal channels. Time is not represented.





Nigam et al.~\cite{Nigam-Esorics2016} 
work around the notion of Timed Dolev-Yao Intruder Models for
Cyber-Physical Security Protocols by bounding the number of intruders
required for the automated verification of such protocols. Following a
tradition in security protocol analysis, they provide an answer to the
question: How many intruders are enough for verification and where should
they be placed? They also extend the strand space model to \CPS{}
protocols by allowing for the symbolic representation of time, so that
they can use the tool Maude~\cite{Maude} along with SMT support. Their
notion of time is however different from ours, as they focus on the time a
message needs to travel from an agent to another. The paper does not
mention physical devices, such as sensors and/or actuators.





\subsubsection*{Future work} While much is still to be done, we believe that our paper 
provides a
stepping stone for the development of formal and automated tools to
analyse the security of \CPS{s}. We will consider
applying, possibly after proper enhancements, existing tools and
frameworks for automated security protocol analysis, resorting to the
development of a dedicated tool if existing ones prove not up to the task.
We will also consider further security properties and concrete examples of
\CPS{s}, as well as other kinds of cyber-physical attackers and attacks,
e.g., periodic attacks. This will allow us to refine the classes of
attacks we have given here (e.g., by formalising a type system amenable to
static analysis), and provide a formal definition of when a \CPS{} is more
secure than another so as to be able to design, by progressive refinement,
secure variants of a vulnerable \CPS{s}.

We also aim to extend the preliminary quantitative analysis we have given
here by developing a suitable behavioural theory ensuring that our trace
semantics considers also the probability of a trace to actually occur. We
expect that the discrete time stochastic hybrid systems of~\cite{abate06}
will be useful to that extent.


\bibliographystyle{abbrv}

\bibliography{IoT_bib}


\appendix

\section{Proofs}

\subsection{Proof of \autoref{sec:calculus}}

In order to prove \autoref{prop:sys} and \autoref{prop:X}, we use the following lemma that formalises the invariant properties binding 
the state variable $\mathit{temp}$ with the activity of the cooling system.

Intuitively, when the cooling system is inactive the value of the state
variable $\mathit{temp}$ lays in the real interval $[0, 11.5]$.
Furthermore, if the coolant is not active and the variable $\mathit{temp}$
lays in the real interval $(10.1, 11.5]$, then the cooling will be turned
on in the next time slot. Finally, when active the cooling system will
remain so for $k\in1..5$ time slots (counting also the current time slot)
with the variable $\mathit{temp}$ being in the real interval $(
9.9-k{*}(1{+}\delta) , 11.5-k{*}(1{-}\delta)]$.

\begin{lemma} 
\label{lem:sys}
Let $\mathit{Sys}$ be the system defined in \autoref{exa:sys}.
Let
\begin{small}
\begin{displaymath}
\mathit{Sys} = \mathit{Sys_1} \trans{t_1}\trans\tick 
\mathit{Sys_2}\trans{t_2}\trans\tick  \dots 
\trans{t_{n-1}}\trans\tick  \mathit{Sys_n}
\end{displaymath}
\end{small}%
such that the traces $t_j$ contain no $\tick$-actions, for any $j \in  1 .. n{-}1 $, and for any  $i \in  1 .. n $, $\mathit{Sys_i}= \confCPS {E_i}{P_i} $ with 
$E_i = \envCPS 
{\statefun^i{}} 
{\actuatorfun^i{}} 
{ \delta }  
{\evolmap{}}
{ \epsilon }  
{\measmap{}}   
{\invariantfun{}}$.
Then, for any $i \in 1 .. n{-}1 $, we have the following:
\begin{enumerate}

\item \label{uno}
 if   $ \actuatorfun^i{}(\mathit{cool})= \off $ then
 $\statefun^i{}(\mathit{temp})  \in [0, 11.1+\delta ]$;
with  $\statefun^i{}(\mathit{stress})=0$ if $ \statefun^i{}(\mathit{temp})  \in [0, 10.9+\delta ] $, and  $\statefun^i{}(\mathit{stress})=1 $, otherwise; 

\item \label{due}
  if   $ \actuatorfun^i{}(\mathit{cool})= \off $ and 
$\statefun^i{}(\mathit{temp})\in (10.1, 11.1+\delta ]$ then, in the next time slot,  $\actuatorfun^{i{+}1}{}(\mathit{cool})=\on$
  and $\statefun^{i{+}1}{}(\mathit{stress})  \in 1..2$;

\item \label{tre}
 if  $ \actuatorfun^i{}(\mathit{cool})=\on$ then   $\statefun^i{}(\mathit{temp}) \in ( 9.9-k {*}(1{+}\delta) , 11.1+\delta -k{*}(1{-}\delta)] $, 
for some  $k  \in 1 .. 5 $   such that $\actuatorfun^{i-k}{}(\mathit{cool})=\off $ and 
$\actuatorfun^{i-j}{}(\mathit{cool}) =\on $, for $j \in 0..k{-}1$;
moreover,  if  $k\in 1.. 3$ then     $\statefun^i{}(\mathit{stress})  \in  1..k{+}1  $, otherwise,  
$\statefun^i{}(\mathit{stress}) =0$. 
\end{enumerate}
\end{lemma}
\begin{proof}
Let us write $v_i$ and $s_i$ to denote the values of the state variables
$\mathit{temp}$ and $\mathit{stress}$, respectively, in the systems
$\mathit{Sys_i}$, i.e., $\statefun^i{} (\mathit{temp})=v_i $ and
$\statefun^i{} (\mathit{stress})=s_i $. Moreover, we will say that the
coolant is active (resp., is not active) in $\mathit{Sys_i}$ if
$\actuatorfun^i{}(\mathit{cool})=\on$ (resp.,
$\actuatorfun^i{}(\mathit{cool})=\off$).

The proof is by mathematical induction on $n$, i.e., the number of
$\tick$-actions of our traces.

The \emph{case base} $n=1$ follows directly from the definition of $\mathit{Sys}$. 

Let us prove the \emph{inductive case}. 
We assume that the three statements hold for $n-1$ and prove that they  
also hold for $n$.
\begin{enumerate}[noitemsep]
\item Let us assume that the cooling  is not active  in $\mathit{Sys_{n}}$.
In this case, we prove that $v_n \in [0, 11.1+\delta ]$, with and $s_n=0$ if $ v_n  \in [0, 10.9+\delta ] $, and $s_n=1$ otherwise.

We consider separately the cases in which the coolant is active or not in $\mathit{Sys_{n-1}}$
\begin{itemize}[noitemsep]
\item Suppose the coolant is not active in $\mathit{Sys_{n{-}1}}$ (and
not active in $\mathit{Sys_{n}}$).

By the induction hypothesis we have 
$v_{n-1} \in [0, 11.1+\delta ]$; with $s_{n{-}1}=0$ if $ v_{n{-}1}  \in [0, 10.9+\delta ] $, and $s_{n{-}1}=1 $ otherwise. Furthermore,  if   
$v_{n-1} \in (10.1, 11.1+\delta ]$, then, by the induction hypothesis, the coolant must be active in $\mathit{Sys_{n}}$.
Since we know that in $\mathit{Sys_n}$ the cooling is not active,
it follows that $v_{n-1} \in [0, 10.1]$ and $s_n =0$.
Furthermore, in $\mathit{Sys_{n}}$ the temperature
will increase of a value laying in the real interval $[1-\delta,1+\delta]=[0.6,1.4]$. Thus, $v_{n}$ will be in 
$ [0.6, 11.1+\delta ]\subseteq[0, 11.1+\delta ]$.  
Moreover, if $v_{n-1} \in [0, 9.9]$, then the state variable $\mathit{stress}$ is not incremented and hence $s_n=0$ with   
$ v_n  \in [0+1-\delta \, , \,  9.9+ 1+\delta ]=[0.6 \, , \, 10.9+\delta]\subseteq [0 \, , \, 10.9+\delta ] $. Otherwise, 
if $v_{n-1} \in (9.9,10.1]$, then the state variable $\mathit{stress}$ is incremented, and hence $s_n=1$.

\item Suppose the coolant is active in $\mathit{Sys_{n{-}1}}$ (and not
active in $\mathit{Sys_{n}}$).

By the induction hypothesis, $v_{n-1} \in ( 9.9-k *(1+\delta) , 11.1+\delta
-k*(1-\delta)]$ for some $k \in 1..5$ such that the coolant is not active
in $\mathit{Sys_{n{-}1{-}k}}$ and is active in $\mathit{Sys_{n{-}k}},
\ldots, \mathit{Sys_{n-1}}$.

The case $k \in \{1,\ldots,4\}$  is not admissible. 
In fact if  $k \in \{1,\ldots,4\}$ then the coolant would be active for less than $5$ $\tick$-actions as we know that 
$\mathit{Sys_{n}}$ is not active. 
Hence, it must be $k=5$. Since $\delta=0.4$ and $k=5$, it holds that $v_{n-1 }\in (9.9-5*1.4, 11.1+0.4 -5*0.6]=(2.8, 8.6] $
and $s_{n{-}1} =0$. Moreover, since
the coolant is active for $5$ time slots, in $\mathit{Sys_{n{-}1}}$ the controller and the $\mathit{IDS}$ synchronise together via channel $\mathit{sync}$ and hence the $\mathit{IDS}$ checks the
temperature. Since $v_{n-1} \in (2.8, 8.6]$ the $\mathit{IDS}$ process 
sends to the controller a command to 
$\mathsf{stop}$ the cooling, and the controller will switch off the 
cooling system. Thus, in the next time slot, the temperature
will increase of a value laying in the real interval $[1-\delta,1+\delta]=[0.6,1.4]$. As
a consequence, in $\mathit{Sys_{n}}$ we will have $v_{n}
\in [2.8+0,6, 8.6+1.4]=[3.4,10] \subseteq [0, 11.1+\delta ]$.
Moreover, since $v_{n-1} \in (2.8, 8.6]$ and $s_{n{-}1} =0$, we derive that the state variable $\mathit{stress}$ is not increased and hence
$s_{n } =0$, with $v_{n} \in [3.4,10] \subseteq [0, 10.9+\delta ]$. 
\end{itemize}

\item Let us assume that the coolant is not active in $\mathit{Sys_{n}}$
and $v_n \in (10.1, 11.1+\delta ]$; we prove that the coolant is active in
$\mathit{Sys_{n{+}1}}$ with $s_{n {+}1} \in 1..2 $. Since the coolant is
not active in $\mathit{Sys_{n}}$, then it will check the temperature
before the next time slot. Since $v_n \in (10.1, 11.1+\delta ]$ and
$\epsilon=0.1$, then the process $\mathit{Ctrl}$ will sense a temperature
greater than $10$ and the coolant will be turned on. Thus, the coolant
will be active in $\mathit{Sys_{n{+}1}}$. Moreover, since $v_n \in (10.1,
11.1+\delta ]$, and $s_{n}$ could be either $0$ or $1$, the state variable
$\mathit{stress}$ is increased and therefore $s_{n {+}1} \in 1..2$.

\item Let us assume that the coolant is active in $\mathit{Sys_{n}}$; we
prove that $v_{n} \in ( 9.9-k *(1+\delta), 11.1+\delta -k*(1-\delta)] $
for some $k \in 1..5 $ and the coolant is not active in
$\mathit{Sys_{n{-}k}}$ and active in $\mathit{Sys_{n-k+1}}, \dots,
\mathit{Sys_{n}}$. Moreover, we have to prove that if $k\leq 3$ then $s_n
\in 1..k{+}1 $, otherwise, if $k > 3$ then $s_n =0$.

We prove the first statement. That is, we prove that $v_{n} \in ( 9.9-k
*(1+\delta), 11.1+\delta -k*(1-\delta)] $, for some $k \in 1..5 $, and the
coolant is not active in $\mathit{Sys_{n{-}k}}$, whereas it is active in
the systems $\mathit{Sys_{n-k+1}}, \dots, \mathit{Sys_{n}}$.


We separate the case in which the coolant is active in
$\mathit{Sys_{n{-}1}}$ from that in which is not active.

\begin{itemize}[noitemsep]
\item Suppose the coolant is not active in $\mathit{Sys_{n{-}1}}$ (and active in $\mathit{Sys_{n}}$).

In this case $k=1$ as the coolant is not active in $\mathit{Sys_{n-1}}$
and it is active in $\mathit{Sys_{n}}$. Since $k=1$, we have to prove $v_n
\in (9.9-(1+\delta), 11.1+\delta-(1-\delta)]$.

However, since the coolant is not active in $\mathit{Sys_{n-1}}$ and is
active in $\mathit{Sys_{n}}$ it means that the coolant has been switched
on in $\mathit{Sys_{n-1}}$ because the sensed temperature was above $10$
(since $\epsilon=0.1$ this may happen only if $v_{n-1} > 9.9$). By the
induction hypothesis, since the coolant is not active in
$\mathit{Sys_{n-1}}$, we have that $v_{n-1} \in [0, 11.1+\delta ]$.
Therefore, from $v_{n-1} > 9.9$ and $v_{n-1} \in [0, 11.1+\delta ]$ it
follows that $v_{n-1} \in (9.9, 11.1+\delta ]$. Furthermore, since the
coolant is active in $\mathit{Sys_{n}}$, the temperature will decrease of
a value in $[1-\delta,1+\delta]$ and therefore $v_n \in (9.9-(1+\delta),
11.1+\delta-(1-\delta)]$, which concludes this case of the proof.

\item Suppose the coolant is active in $\mathit{Sys_{n{-}1}}$ (and active
in $\mathit{Sys_{n}}$ as well).

By the induction hypothesis, there is $h \in 1..5$ such that $v_{n-1} \in
( 9.9-h *(1+\delta) , 11.1+\delta -h*(1-\delta)] $ and the coolant is not
active in $\mathit{Sys_{n{-}1{-}h}}$ and is active in
$\mathit{Sys_{n{-}h}}, \ldots, \mathit{Sys_{n{-}1}}$.

The case $h=5$ is not admissible. In fact, since $\delta=0.4$, if $h=5$
then $v_{n-1 }\in (9.9-5*1.4, 11.1+\delta -5*0.6]=(2.8, 8.6] $.
Furthermore, since the cooling system has been active for $5$ time
instants, in $\mathit{Sys_{n{-}1}}$ the controller and the IDS synchronise
together via channel $\mathit{sync}$, and the $\mathit{IDS}$ checks the
received temperature. As $v_{n-1 }\in (2.8, 8.6] $, the $\mathit{IDS}$
sends to the controller via channel $\mathit{ins}$ the command
$\mathsf{stop}$. This implies that the controller should turn off the
cooling system, in contradiction with the hypothesis that the coolant is
active in $\mathit{Sys_{n }}$.

Hence, it must be $h \in 1 .. 4$. Let us prove that for $k=h+1$ we obtain
our result. Namely, we have to prove that, for $k=h+1$, (i) $v_{n} \in (
9.9-k *(1+\delta), 11.1+\delta -k*(1-\delta)] $, and (ii) the coolant is
not active in $\mathit{Sys_{n{-}k}}$ and active in $\mathit{Sys_{n-k+1}},
\dots, \mathit{Sys_{n}}$.

Let us prove the statement (i). By the induction hypothesis, it holds that
$v_{n-1} \in ( 9.9-h *(1+\delta) , 11.1+\delta -h*(1-\delta)] $. Since the
coolant is active in $\mathit{Sys_{n}}$, the temperature will decrease
Hence, $v_{n } \in ( 9.9-(h+1) *(1+\delta) , 11.1+\delta
-(h+1)*(1-\delta)] $. Therefore, since $k=h+1$, we have that $v_{n} \in (
9.9-k *(1+\delta) , 11.1+\delta -k*(1-\delta)] $.

Let us prove the statement (ii). By the induction hypothesis the coolant
is not active in $\mathit{Sys_{n-1-h}}$ and it is active in
$\mathit{Sys_{n-h}}, \ldots, \mathit{Sys_{n-1}}$. Now, since the coolant
is active in $\mathit{Sys_{n}}$, for $k=h+1$, we have that the coolant is
not active in $\mathit{Sys_{n-k}}$ and is active in $\mathit{Sys_{n-k+1}},
\ldots, \mathit{Sys_{n}}$, which concludes this case of the proof.
\end{itemize}

Thus, we have proved that $v_{n} \in ( 9.9-k *(1+\delta), 11.1+\delta
-k*(1-\delta)] $, for some $k \in 1..5 $; moreover, the coolant is not
active in $\mathit{Sys_{n{-}k}}$ and active in the systems
$\mathit{Sys_{n-k+1}}, \dots, \mathit{Sys_{n}}$.

It remains to prove that $s_n \in 1..k{+}1 $ if $k\leq 3$, and $s_n =0$,
otherwise.

By inductive hypothesis, since the coolant is not active in
$\mathit{Sys_{n{-}k}}$, we have that $s_{n{-}k} \in 0..1$. Now, for $k \in
[1..2]$, the temperature could be greater than $9.9$. Hence if the state
variable $\mathit{stress}$ is either increased or reset, then $s_n \in
1..k{+}1 $, for $k\in 1.. 3$. Moreover, since for $k\in 3..5 $ the
temperature is below $9.9$, it follows that $s_n =0$ for $k> 3$.
\end{enumerate}
\end{proof}

\begin{proof}[Proof of  \autoref{prop:sys}]
Since $\delta=0.4$, by \autoref{lem:sys} the value of the state variable
$\mathit{temp}$ is always in the real interval $[0, 11.5]$. As a
consequence, the invariant of the system is never violated and the system
never deadlocks. Moreover, after $5$ time units of cooling, the state
variable $\mathit{temp}$ is always in the real interval $( 9.9-5 *1.4 ,
11.1+0.4-5*0.6]=(2.9, 8.5]$. Hence, the process $\mathit{IDS}$ will never
transmit on the channel $\mathit{alarm}$.

Finally, by \autoref{lem:sys} the maximum value reached by the state variable $\mathit{stress}$ is $4$ and therefore the system does not reach unsafe states.
\end{proof}

\begin{proof}[Proof of  \autoref{prop:X}]
Let us prove the two statements separately. 
\begin{itemize}
\item Since $\epsilon=0.1$, if process $\mathit{Ctrl}$ senses a
temperature above $10$ (and hence $\mathit{Sys}$ turns on the cooling)
then the value of the state variable $\mathit{temp}$ is greater than
$9.9$. By \autoref{lem:sys}, the value of the state variable
$\mathit{temp}$ is always less than or equal to $11.1+\delta $. Therefore, if
$\mathit{Ctrl}$ senses a temperature above $10$, then the value of the
state variable $\mathit{temp}$ is in $(9.9,11.1+\delta ]$.

\item By \autoref{lem:sys} (third item), the coolant can be active for no
more than $5$ time slots. Hence, by \autoref{lem:sys}, when $\mathit{Sys}$
turns off the cooling system the state variable $\mathit{temp}$ ranges
over $( 9.9-5 *(1+\delta) , 11.1+\delta-5*(1-\delta)]$. 
\end{itemize}
\end{proof}

\subsection{Proofs of \autoref{sec:cyber-physical-attackers}}

%
%
%

\begin{proof}[Proof of \autoref{prop:att:DoS}]
We distinguish the two cases, depending on $m$.
\begin{itemize}[noitemsep]
\item Let $m\leq 8$.
We recall that the cooling system is activated only when the sensed temperature is above $10$. Since $\epsilon = 0.1$, when this happens the state variable $\mathit{temp}$ must be at least $9.9$. Note that after $m{-}1 \leq 7$ $\tick$-actions, when the attack tries to interact with the
controller of the actuator $\mathit{cool}$, the variable $\mathit{temp}$ may reach at most $7* (1 + \delta)= 7 * 1.4=9.8$ degrees. Thus, the cooling system will not be activated and the attack will not have any effect.

\item Let $m>8$. 
By \autoref{prop:sys}, the system $\mathit{Sys}$ in
isolation may never deadlock, it does not get into an unsafe state, and it may never emit an output on channel $\mathit{alarm}$. Thus, any execution trace of the system $\mathit{Sys}$ consists of a sequence of $\tau$-actions and $\tick$-actions.

In order to prove the statement it is enough to show the following four facts:
\begin{itemize}[noitemsep]
\item the system $\mathit{Sys} \parallel A_m$ may not deadlock in the first
$m+3$ time slots; 

\item the system $\mathit{Sys} \parallel A_m$ may not emit any output in the first $m+3$ time slots; 

\item the system $\mathit{Sys} \parallel A_m$ may not enter in an unsafe
state in the first $m+3$ time slots;

\item the system $\mathit{Sys} \parallel A_m$ has a trace reaching un unsafe state from the $(m{+}4)$-th time slot on, and until the invariant gets violated and the system deadlocks. 
\end{itemize}


The first three facts are easy to show as the attack may steal the command
addressed to the actuator $\mathit{cool}$ only in the $m$-th time slot.
Thus, until time slot $m$, the whole system behaves correctly. In
particular, by \autoref{prop:sys} and \autoref{prop:X}, no alarms,
deadlocks or violations of safety conditions occur, and the temperature
lies in the expected ranges. Any of those three actions requires at least
further $4$ time slots to occur. Indeed, by \autoref{lem:sys}, when the
cooling is switched on in the time slot $m$, the variable
$\mathit{stress}$ might be equal to $2$ and hence the system might not
enters in an unsafe state in the first $m+3$ time slots. Moreover, an
alarm or a deadlock needs more than $3$ time slots and hence no alarm can
occur in the first $m+3$ time slots.

Let us show the fourth fact, i.e., that there is a trace where the system
$\mathit{Sys} \parallel A_m$ enters into an unsafe state starting from the
$(m{+}4)$-th time slot and until the invariant gets violated.

Firstly, we prove that for all time slots $n$, with $9\leq n < m$, there
is a trace of the system $\mathit{Sys} \parallel A_m$ in which the state
variable $\mathit{temp}$ reaches the values $10.1$ in the time slot $n$.

The fastest trace reaching the temperature of $10.1$ degrees requires
$\lceil \frac{10.1}{1 + \delta }\rceil = \lceil \frac{10.1}{1.4 }\rceil
=8$ time units, whereas the slowest one $\lceil \frac{10.1}{1 - \delta
}\rceil = \lceil \frac{10.1}{0.6 }\rceil =17$ time units. Thus, for any
time slot $n$, with $9 \leq n \leq 18$, there is a trace of the system
where the value of the state variable $\mathit{temp}$ is $10.1$. Now, for
any of those time slots $n$ there is a trace in which the state variable
$\mathit{temp}$ is equal to $10.1$ in all time slots $n+10i < m$, with
$i\in \mathbb{N}$. Indeed, when the variable $\mathit{temp}$ is equal to
$10.1$ the cooling might be activated. Thus, there is a trace in which the
cooling system is activated. We can always assume that during the cooling
the temperature decreases of $1+\delta$ degrees per time unit, reaching at
the end of the cooling cycle the value of $5$. This entails that the trace
may continue with $5$ time slots in which the variable $\mathit{temp}$ is
increased of $1+\delta$ degrees per time unit; reaching again the value
$10.1$. Thus, for all time slots $n$, with $9 \leq n < m$, there is a
trace of the system $\mathit{Sys} \parallel A_m$ in which the state
variable $\mathit{temp}$ is $10.1$ in $n$.

As a consequence, we can suppose that in the $m{-}1$-th time slot there is
a trace in which the value of the variable $\mathit{temp}$ is $10.1$.
Since $\epsilon=0.1$, the sensed temperature lays in the real interval
$[10,10.2]$. Let us focus on the trace in which the sensed temperature is
$10$ and the cooling system is not activated. In this case, in the $m$-th
time slot the system may reach a temperature of $10.1 + (1 + \delta)=11.5$
degrees and the variable $\mathit{stress}$ is $1$.

The process $\mathit{Ctrl}$ will sense a temperature above $10$ sending
the command $\snda {cool} {\on}$ to the actuator $\mathit{cool}$. Now,
since the attack $A_m$ is active in that time slot ($m > 8$), the command
will be stolen by the attack and it will never reach the actuator. Without
that dose of coolant, the temperature of the system will continue to grow.
As a consequence, after further $4$ time units of cooling, i.e.\ in the
$m{+}4$-th time slot, the value of the state variable $\mathit{stress}$
may be $5$ and the system enters in an unsafe state.

After $1$ time slots, in the time slot $m+5$, the controller and the
$\mathit{IDS}$ synchronise via channel $\mathit{sync}$, the $\mathit{IDS}$
will detect a temperature above $10$, and it will fire the output on
channel $\mathit{alarm}$ saying to process $\mathit{Ctrl}$ to keep
cooling. But $\mathit{Ctrl}$ will not send again the command $\snda {cool}
{\on}$. Hence, the temperature would continue to increase and the system
remains in an unsafe state while the process $\mathit{IDS}$ will keep
sending of $\mathit{alarm}$(s) until the invariant of the environment gets
violated.
\end{itemize}
\end{proof}

\begin{proof}[Proof of \autoref{prop:critical2}]
By induction on the length of the trace.  
\end{proof}

In order to prove \autoref{prop:att:integrity}, we introduce \autoref{lem:sys2}. This is a variant of \autoref{lem:sys} in which the \CPS{} $\mathit{Sys} $ runs in parallel with the attack $A_n$ defined in \autoref{exa:att:integrity}. Here, due to the presence of the attack, the temperature is $2$ degrees higher when compared to the system $\mathit{Sys}$ in isolation. 
\begin{lemma} 
\label{lem:sys2}
Let $\mathit{Sys}$ be the system defined in \autoref{exa:sys} and $A_n$ be the attack of \autoref{exa:att:integrity}. 
Let
\begin{small}
\begin{displaymath}
\mathit{Sys} \parallel A_n= \mathit{Sys_1}  \trans{t_1}\trans\tick \dots \mathit{Sys_{n-1}}  \trans{t_{n-1}}\trans\tick \mathit{Sys_n} 
\end{displaymath}
\end{small}%
such that the traces $t_j$ contain no $\tick$-actions, for any $j \in  1 .. n{-}1 $,  and for any  $i \in  1 .. n $ $\mathit{Sys_i}= \confCPS {E_i}{P_i} $ with 
$E_i = \envCPS 
{\statefun^i{}} 
{\actuatorfun^i{}} 
{ \delta }  
{\evolmap{}}
{ \epsilon }  
{\measmap{}}   
{\invariantfun{}}$.
Then, for any $i \in 1 .. n{-}1 $ we have the following:
\begin{itemize}[noitemsep]
\item  if   $ \actuatorfun^i{}(\mathit{cool})= \off $, then
 $\statefun^i{}(\mathit{temp})\in [0, 11.1+2+\delta ]$; 

\item if $ \actuatorfun^i{}(\mathit{cool})= \off $ and
$\statefun^i{}(\mathit{temp})\in (10.1+2, 11.1+2+\delta ]$, then we have $
\actuatorfun^{i+1}{}(\mathit{cool}) =\on$;

\item  if  $ \actuatorfun^i{}(\mathit{cool})=\on$, then   $\statefun^i{}(\mathit{temp}) \in ( 9.9+2-k *(1+\delta) , 11.1+2+\delta -k*(1-\delta)] $, 
for some  $k  \in 1..5$,   such that $\actuatorfun^{i-k}{}(\mathit{cool}) =\off $ and $ \actuatorfun^{i-j}{}(\mathit{cool}) =\on $, for $j \in 0..k{-}1$. 

\end{itemize}
\end{lemma}
\begin{proof}
Similar to the proof of    \autoref{lem:sys}.
\end{proof}

Now, everything is in place to prove \autoref{prop:att:integrity}. 
\begin{proof}[Proof of  \autoref{prop:att:integrity}]
Let us proceed by case analysis. 
\begin{itemize}[noitemsep]
\item 
Let $0 \leq n \leq 8$.  
In the proof of \autoref{prop:att:DoS}, we remarked that the system $\mathit{Sys}$ in isolation may sense a temperature greater than $10$  
only after $8$ $\tick$-actions, i.e., in the $9$-th time slot.
However, the life of the attack is $n \leq 8$, and in the $9$-th time
slot the attack is already terminated. As a consequence, starting
from the $9$-th time slot the system will correctly sense the
temperature and it will correctly activate the cooling system.
\item Let $n=9$. 
The maximum value that may be reached by the state variable
$\mathit{temp}$ after $8$ $\tick$-actions, i.e., in the $9$-th time
slot, is $8 * (1+ \delta)=8*1.4= 11.2$. However, since in the $9$-th
time slot the attack is still alive, the process $\mathit{Ctrl}$
will sense a temperature below $10$ and the system will move to the next
time slot and the state variable   $\mathit{stress}$ is incremented. Then, in the $10$-th time slot, when the attack is
already terminated, the maximum temperature the system may reach is $11.2
+ (1+ \delta)=12.6$ degrees and the state variable $\mathit{stress}$ is equal to  $1$. Thus, the process $\mathit{Ctrl}$ will sense
a temperature greater than $10$, activating the cooling system  and incrementing the state variable $\mathit{stress}$. 
As a consequence, during the following $4$ time units of cooling, the value of the state variable $\mathit{temp}$ will be at most
$12.6 - 4*(1-\delta)= 12.6- 4*0.6=10.2$, and hence  in the $14$-th time
slot, the value of the state variable $\mathit{stress}$ is $5$. 
As a consequence,  the system will enter in an unsafe state.
In the $15$-th time
slot, the value of the state variable $\mathit{stress}$ is still equal to $5$ 
and the system will still be in an unsafe state.
However, 
 the value of the state variable $\mathit{temp}$ will be at most
$12.6 - 5*(1-\delta)= 12.6- 5*0.6=9.6$
 which will be sensed by process $\mathit{IDS}$  as at most
$9.7$ (sensor error $\epsilon=0.1$). As a consequence,  no alarm will be turned on and the   variable $\mathit{stress}$ will be reset.
Moreover, the invariant will be obviously always preserved. 

As in the current time slot the attack has already terminated, from this
point in time on, the system will behave correctly with neither deadlocks
or alarms.

\item 
Let $n \geq 10$. 
In order to prove that  $\mathit{Sys} \parallel A_{n} \, \simeq_{[14,n{+}7]} \,
\mathit{Sys}$,   it is enough to show that:
\begin{itemize}[noitemsep]

\item  the system $\mathit{Sys} \parallel A_n$ does not deadlock;

\item the system $\mathit{Sys} \parallel A_{n}$ may not emit any output in the first $13$ time slots; 
\item there is a trace in which the system $\mathit{Sys} \parallel A_{n}$ enters in an unsafe state in the $14$-th time slot;

\item 
there is a trace in which the system $\mathit{Sys} \parallel A_n$ is in an unsafe state in the $(n{+}7)$-th time slot;

 \item   the system $\mathit{Sys} \parallel A_n$ does not have any execution
trace emitting an output along channel $\mathit{alarm}$
 or entering in an unsafe state after   the $n+7$-th time slot. 

\end{itemize}

As regards the first fact, since $\delta=0.4$, by \autoref{lem:sys2} the
temperature of the system under attack will always remain in the real
interval $[0, 13.5]$. Thus, the invariant is never violated and the trace
of the system under attack cannot contain any $\dead$-action. Moreover,
when the attack terminates, if the temperature is in $[0,9.9]$, the system
will continue his behaviour correctly, as in isolation. Otherwise, since
the temperature is at most $13.5$, after a possible sequence of cooling
cycles, the temperature will reach a value in the interval $[0,9.9]$, and
again the system will continue its behaviour correctly, as in isolation.

Concerning the second and the third facts, the proof is analogous to that
of case $ n=9$.


Concerning the fourth fact, firstly we prove that for all time slots $m$,
with $9 < m \leq n$, there is a trace of the system $\mathit{Sys}
\parallel A_n$ in which the state variable $\mathit{temp}$ reaches the
values $12$ in the time slot $m$. Since the attack is alive at that time,
and $\epsilon=0.1$, when the variable $\mathit{temp}$ will be equal to
$12$ the sensed temperature will lay in the real interval $[9.9,10.1]$.

The fastest trace reaching the temperature of $12$ degrees requires
$\lceil \frac{12}{1 + \delta }\rceil = \lceil \frac{12}{1.4 }\rceil =9$
time units, whereas the slowest one $\lceil \frac{12}{1 - \delta }\rceil =
\lceil \frac{12}{0.6 }\rceil =20$ time units. Thus, for any time slot $m$,
with $9 < m \leq 21$, there is a trace of the system where the value of
the state variable $\mathit{temp}$ is $12$. Now, for any of those time
slots $m$ there is a trace in which the state variable $\mathit{temp}$ is
equal to $12$ in all time slots $m+10i < n$, with $i\in \mathbb{N}$. As
already said, when the variable $\mathit{temp}$ is equal to $12$ the
sensed temperature lays in the real interval $[9.9, 10.1]$ and the cooling
might be activated. Thus, there is a trace in which the cooling system is
activated. We can always find a trace where during the cooling the
temperature decreases of $1+\delta$ degrees per time unit, reaching at the
end of the cooling cycle the value of $5$. Thus, the trace may continue
with $5$ time slots in which the variable $\mathit{temp}$ is increased of
$1+\delta$ degrees per time unit; reaching again the value $12$. Thus, for
all time slots $m$, with $9 < m \leq n$, there is a trace of the system
$\mathit{Sys} \parallel A_n$ in which the state variable $\mathit{temp}$
has value $12$ in the time slot $m$.

Therefore, we can suppose that in the $n$-th time slot the variable
$\mathit{temp}$ is equal to $12$ and, since the maximum increment of
temperature is $1.4$, the the variable $\mathit{stress}$ is at least equal
to $1$. Since the attack is alive and $\epsilon=0.1$, in the $n$-th time
slot the sensed temperature will lay in $[9.9,10.1]$. We consider the case
in which the sensed temperature is less than $10$ and hence the cooling is
not activated.

Thus, in the $n{+}1$-th time slot the system may reach a temperature of
$12 + 1 + \delta=13.4$ degrees and the process $\mathit{Ctrl}$ will sense
a temperature above $10$, and it will activate the cooling system. In this
case, the variable $\mathit{stress}$ will be increased. As a consequence,
after further $5$ time units of cooling, i.e.\ in the $n{+}6$-th time
slot, the value of the state variable $\mathit{temp}$ may reach $13.5 -
5*(1-\delta)=10.4$ and the alarm will be fired and the variable
$\mathit{stress}$ will be still equal to $5$. Therefore, in the $n{+}7$-th
time slot the variable $\mathit{stress}$ will be still equal to $5$ and
the system will be in an unsafe state.

Concerning the fifth fact,
%
%
by \autoref{lem:sys2}, in the $n{+}1$-th time slot the attack will be terminated and the system may reach a temperature that is, in the worst case, at most $13.5$. Thus, the cooling system may be activated and the variable $\mathit{stress}$ will be increased. As a consequence, in the $n{+}7$-th time slot, the value of the state variable $\mathit{temp}$ may be at most $13.5-6*(1-\delta)=13.5-6*0.6=9.9$ and the variable $\mathit{stress}$ will be reset to $0$. Thus, after the $n+7$-th time slot, the system will behave correctly, as in isolation.
\end{itemize}
\end{proof}

In order to prove \autoref{thm:sound}, we introduce the following lemma. 
\begin{lemma}
\label{lem:top}
Let $M$ be an honest and sound  \CPS{}, $C$ a class
of attacks, and 
$A$  an  attack of an arbitrary class  $C' \preceq C$.
Whenever $M\parallel A\trans t  M'\parallel A'$, then 
\[
M\parallel \mathit{Top}(C) \trans t  M'\parallel\prod_{\iota \in \I } \mathit{Att}( 
\iota ,  \#\tick(t){+}1   , C(\iota)) \enspace .
\]
\end{lemma}

\begin{proof}
Let us denote with $\mathit{Top}^h(C)$ the attack process
\[
  \prod_{\iota \in \I } \mathit{Att}( 
\iota , h , C(\iota)).
\]
Obviously, $\mathit{Top}^1 (C)=\mathit{Top} (C)$.

The proof is by mathematical induction on the  length $k$ of the trace $t$.

\noindent
\emph{Base case} $k=1$. \\
This means $t=\alpha $, for some action $\alpha$.
We proceed by case analysis on the action $\alpha$.
\begin{itemize}
\item $\alpha = \out {c} v$. Since the attacker $A$ does not use a
communication channel, from $M\parallel A\trans {\out {c} v} M'\parallel
A'$ we can derive that $A=A'$ and $M \trans {\out {c} v} M' $. Hence by
rules \rulename{Par} and \rulename{Out}, we derive $M \parallel
\mathit{Top}(C) \trans {\out {c} v} M' \parallel \mathit{Top}^1(C)=M'
\parallel \mathit{Top}(C)$.

\item $\alpha = \inp {c} v$. This case is similar to the previous one.

\item $\alpha =\tau$. There are several sub-cases. 
\begin{itemize}
\item Let $M\parallel A\trans {\tau} M'\parallel A'$ be derived by an
application of rule \rulename{SensReadSec}. Since the attacker $A$
performs only malicious actions on physical devices, from $M\parallel
A\trans {\tau} M'\parallel A'$ we can derive that $A=A'$ and $P \trans
{\rcva { s} v} P' $, for some processes $P$ and $P'$ such that
$M=\confCPSS E {\cal S} P$ and $M'=\confCPSS E {\cal S} {P'}$. Hence by an
application of rules \rulename{Par} and \rulename{SensReadSec} we derive
$M \parallel \mathit{Top}(C) \trans {\tau} M' \parallel
\mathit{Top}^1(C)=M' \parallel \mathit{Top}(C)$.
 
\item  Let $M\parallel A\trans {\tau}  M'\parallel A'$  be derived by
an application of  rule \rulename{ActWriteSec}.
This case is similar to the previous one.

\item  Let $M\parallel A\trans {\tau}  M'\parallel A'$ be  derived by 
an application of rule \rulename{SensReadUnSec}. 
Since the attacker   $A$  performs only malicious actions, from
 $M\parallel A\trans {\tau}  M'\parallel A'$ we can derive that $A=A'$ and
$P \trans {\rcva { s} v}  P' $ for some processes $P$ and $P$' such that
$M=\confCPSS E {\cal S} P$ and $M'=\confCPSS E {\cal S} {P'}$.

By considering $\mathit{rnd}(\{\true,\false\})=\false$ for any  process $\mathit{Att}( 
\iota , 1 , C(\iota))$,  we have that $\mathit{Top}(C)$ can perform only a $\tick $ action, and 
\[ \mathit{Top}(C) \ntrans { \snda {\mbox{\Lightning}s} v} \enspace . \]
 Hence by an application of  rules \rulename{Par} and \rulename{SensReadUnSec}
we derive  $M \parallel \mathit{Top}(C) \trans  {\tau} M' \parallel \mathit{Top}^1(C)=M' \parallel \mathit{Top}(C)$.
 
\item  Let $M\parallel A\trans {\tau}  M'\parallel A'$ be derived by an 
application of rule \rulename{ActWriteUnSec}. 
This case is similar to the previous one.

\item  Let $M\parallel A\trans {\tau}  M'\parallel A'$ be derived by 
an application of  rule \rulename{$\mbox{\Lightning}$SensRead$\mbox{\,\Lightning}$}. 
Since $M$ is sound it follows that  $M=M'$ and  
  $A\trans  {\rcva {\mbox{\Lightning}s} v }A'$. This entails $1  \in C'({\mbox{\Lightning}s?}) \subseteq   C({\mbox{\Lightning}s?}) $, and 
\[ \mathit{Top}(C) \trans { \rcva {\mbox{\Lightning}s} v} \mathit{Top}^1(C) =\mathit{Top}(C) \]  by assuming 
$\mathit{rnd}(\{\true,\false\})=\true$ for the process  
   $ \mathit{Att}(  {\mbox{\Lightning}s?}, 1 , C( {\mbox{\Lightning}s?}))$. 
Hence, by an application of  rules \rulename{Par} and \rulename{$\mbox{\Lightning}$SensRead$\mbox{\,\Lightning}$}
we derive  $M \parallel \mathit{Top}(C) \trans  {\tau} M' \parallel \mathit{Top}^1(C)=M' \parallel \mathit{Top}(C)$.
 
\item  Let $M\parallel A\trans {\tau}  M'\parallel A'$ be derived by 
an application of rule \rulename{$\mbox{\Lightning}$ActWrite$\mbox{\,\Lightning}$}. 
Since $M$ is sound it follows that $M=M'$ and 
 $ A\trans  {\snda {\mbox{\Lightning}a} v }A'$. 
As a consequence,  $1  \in C'({\mbox{\Lightning}a}!) \subseteq C({\mbox{\Lightning}a}!)  $, and 
   \[ \mathit{Top}(C) \trans { \snda {\mbox{\Lightning}a} v} \mathit{Top}^1(C) =\mathit{Top}(C) \]  by assuming 
 $\mathit{rnd}(\{\true,\false\}){=}\true$ and $\mathit{rnd}(\mathbb{R}){=}v$ for the process
    $ \mathit{Att}( {\mbox{\Lightning}a}! , 1 , C({\mbox{\Lightning}a}!))$. 
 Thus, by an application of rules  \rulename{Par} and \rulename{$\mbox{\Lightning}$ActWrite$\mbox{\,\Lightning}$}
we derive  $M \parallel \mathit{Top}(C) \trans  {\tau} M' \parallel \mathit{Top}^1(C)=M' \parallel \mathit{Top}(C)$.

\item  Let $M\parallel A\trans {\tau}  M'\parallel A'$ be derived by 
an application of rule \rulename{Tau}. 
Let $M=\confCPSS E {\cal S} P$ and $M'=\confCPSS {E'} {\cal S} {P'}$. 
There are two possibilities:  either (i)
$P \parallel A \trans \tau P' \parallel A'$, or (ii)
$P \parallel A \trans {\tau :p} P' \parallel A'$.

In the case (i), by inspection of \autoref{tab:lts_processes} and by definition of attacker, it follows that $A$  
cannot perform $\tau$-action since $A$ does not use channel communication
and   performs only malicious actions.
Hence $ P   \trans \tau P' $ and,  by an application 
of rules \rulename{Par} and \rulename{Tau},
we derive  $M \parallel \mathit{Top}(C) \trans  {\tau} M' \parallel \mathit{Top}^1(C)=M' \parallel \mathit{Top}(C)$.

In the case (ii), since $M$ is sound and $A$ can   performs only
 malicious actions,
we have that either (i) $P\trans  {\rcva { s} v }P'$ and $A\trans  {\snda {\mbox{\Lightning}s} v }A'$ 
or, (ii)
 $P\trans  {\snda { a} v }P'$ and $A\trans  {\rcva {\mbox{\Lightning}a} v }A'$.
 We consider the case (i) $P\trans  {\rcva { s} v }P'$ and $A\trans  {\snda {\mbox{\Lightning}s} v }A'$; the case (ii)
is similar.  
Since    
  $A\trans  {\snda {\mbox{\Lightning}s} v }A'$, we derive $1  \in C'(\mbox{\Lightning}s!) \subseteq C(\mbox{\Lightning}s!)  $, and 
   \[ \mathit{Top}(C) \trans { \snda {\mbox{\Lightning}s} v} \mathit{Top}^1(C) =\mathit{Top}(C) \]  by assuming 
 $\mathit{rnd}(\{\true,\false\}){=}\true$ and $\mathit{rnd}(\mathbb{R}){=}v$
   for the process  $ \mathit{Att}( {\mbox{\Lightning}s}! , 1 , C({\mbox{\Lightning}s}!))$. 
  Thus,  by an application of rules  \rulename{$\mbox{\Lightning}$SensWrite$\mbox{\,\Lightning}$}
 and \rulename{Tau}
we derive  $M \parallel \mathit{Top}(C) \trans  {\tau} M' \parallel \mathit{Top}^1(C)=M' \parallel \mathit{Top}(C)$.
 
\end{itemize} 
\item $\alpha = \tick$. This implies that 
 the transition $M\parallel A\trans \tick  M'\parallel A'$ is derived by 
an application of rule  \rulename{Time}.
From $M \parallel A\trans \tick M'\parallel A'$ we derive $M\trans \tick M'$. 
Hence, it suffices to prove that 
$\mathit{Top}(C) \trans \tick \mathit{Top}^2(C)  $  and $M\parallel \mathit{Top}(C) \ntrans \tau$.

First, let us prove that $Top(C) \trans \tick Top^2(C)  $. 
We consider two cases: $1 \in C(\iota)$ and $1 \not\in C(\iota)$.
Let $1 \in C(\iota)$. The transition 
$  \mathit{Att}( \iota, 1 , C(\iota))  \trans \tick   \mathit{Att}( \iota, 2 , C(\iota)) $
 can be derived by assuming    $\mathit{rnd}(\{\true,\false\})=\false$. 
Moreover, since $\mathit{rnd}(\{\true,\false\})=\false$ the process $ \mathit{Att}( \iota, 1 , C(\iota)) $ can perform only a $\tick $ action.
Let   $1 \not \in C(\iota)$. Also in this case the  
process  $ \mathit{Att}( \iota, 1 , C(\iota)) $ can perform only a $\tick $ action. As a consequence, e $  \mathit{Att}( \iota, 1 , C(\iota))  \trans \tick   \mathit{Att}( \iota, 2 , C(\iota)) $. Thus, 
\[ \mathit{Top}(C) \trans \tick  \mathit{Top}^2(C) \enspace . \]

Let us prove  now  that $M\parallel \mathit{Top}(C) \ntrans \tau$.
Since $M\parallel A \ntrans \tau$ it follows that
$M  \ntrans \tau$.
Moreover, since   $\mathit{Top}(C)$ can perform only a $\tick  $ action then, 
by definition of rule \rulename{Time}, it 
follows that $M\parallel \mathit{Top}(C) \ntrans \tau$.

\item $\alpha = \dead$. This case is not possible, because 
 $M\parallel A\trans {\dead}  M'\parallel A'$ would entail 
$M \trans {\dead}  M' $. But $M$ is sound and it cannot deadlock. 

\item $\alpha = \unsafe$. Again, this case is not possible because 
$M$ is sound. 

\end{itemize}

\noindent 
\emph{Inductive case}: $k>1$.\\
We have to prove that
 $M\parallel A\trans t  M'\parallel A'$ implies 
$M\parallel \mathit{Top}(C) \trans t  M'\parallel \mathit{Top}^{ \#\tick(t)+1  }(C)$.

Since the length of $t$ is greater than $1$,
it follows that $t=t' \alpha$, for some $t'$ and $\alpha$.
Hence, there are $M''$ and $A''$ such that 
 \[ M\parallel A\trans {t'}   M''\parallel A'' \trans \alpha  M'\parallel A'
\enspace . \]
By the induction hypothesis, it follows that 
$M\parallel \mathit{Top}(C) \trans { t'}  M''\parallel\mathit{Top}^{\#\tick(t')+1} (C)$. To get the result it is enough to show that 
 $ M''\parallel A'' \trans \alpha  M'\parallel A'$ implies
$M'' \parallel  \mathit{Top}^{\#\tick(t' )+1 }   (C)
\trans {\alpha}  M'\parallel \mathit{Top}^{\#\tick(t  )+1 }(C)$.
The reasoning is similar to that followed in  the \emph{base case\/}, except 
for $\alpha=\dead$ and $\alpha=\unsafe$.
We prove the case $\alpha=\dead$, the other is similar.


Let $M=\confCPSS E {\cal S} P$.
The transition 
   $M''\parallel A\trans {\dead}  M'\parallel A'$ must be derived by
an application of rule  \rulename{Deadlock}. This implies that 
that
$M''=M'$,  $A''=A'$ and $ \invariantfun{}(E)=\false$. 
Thus, by an application of rule \rulename{Deadlock} 
we derive  
\[M'' \parallel  \mathit{Top}^{\#\tick(t' )+1 }   (C)
\trans {\dead}  M'\parallel \mathit{Top}^{\#\tick(t' )+1 }(C) . \]
Since $\#\tick(t  )+1=\#\tick(t'  )+\#\tick(\dead) +1=\#\tick(t' )+1$ we have that
$M'' \parallel  \mathit{Top}^{\#\tick(t' )+1 }   (C)
\trans {\dead}  M'\parallel \mathit{Top}^{\#\tick(t  )+1 }(C)$. 
As required.  
\end{proof}

Everything is finally in place to prove \autoref{thm:sound}.
\begin{proof}[Proof of \autoref{thm:sound}]
The top attacker $\mathit{Top}(C)$ can mimic any execution 
trace of any attack $A$ of class $C'$, with $C' \preceq C$. Thus, by \autoref{lem:top},
if  $M \parallel A\trans t$, for some trace $t$, 
then  $ M \parallel \mathit{Top}(C)\trans t$ as well.

 For any 
$M$ and $A$, either $M \parallel A \sqsubseteq M$ or 
 $M \parallel A   \sqsubseteq_{{m_2}..{n_2}}   M$, 
for some $m_2$ and $n_2$ ($m_2=1$ and 
$n_2=\infty$ if the two systems are completely unrelated).
Suppose by contradiction that $M \parallel A     \not\sqsubseteq   M$ and 
 $M \parallel A    \sqsubseteq_{{m_2}..{n_2}}   M$, 
with  $m_2..n_2 \not \subseteq m_1 .. n_1$. 
There are two cases: either $n_1=\infty$ or $n_1 \in \mathbb{N}^+$. 

If $n_1=\infty$ then $m_2 < m_1$. Since $M \parallel A
\sqsubseteq_{{m_2}..{n_2}} M$, by \autoref{Time-bounded-trace-equivalence}
there is a trace $t$, with $\#\tick(t)=m_2{-} 1$, such that $M \parallel A
\trans t$ and $M \not\!\!\Trans{\hat{t}}$. By \autoref{lem:top}, this
entails $ M \parallel \mathit{Top}(C)\trans t$. Since $M
\not\!\!\Trans{\hat{t}}$ and $\#\tick(t)=m_2{-}1 < m_1$, this contradicts
$M \parallel \mathit{Top}(C) \sqsubseteq_{{m_1}..{n_1}} M$.

If $n_1 \in \mathbb{N}^+$ then $m_2 < m_1$ and/or $n_1 < n_2$, and we
reason as in the previous case.
\end{proof}

\subsection{Proofs of \autoref{sec:impact}}

In order to prove \autoref{prop:toll}, we need a couple of lemmas. 

\autoref{lem:sys3} is a variant of \autoref{lem:sys}. Here the behaviour
of $\mathit{Sys} $ is parametric on the uncertainty.

\begin{lemma} 
\label{lem:sys3}
Let $\mathit{Sys}$ be the system defined in \autoref{exa:sys}, and 
  $0.4 < \gamma \leq \frac{9}{20}$. 
Let
\begin{small}
\begin{displaymath}
\replaceENV {Sys} \delta  \gamma   =\mathit{Sys_1} \trans{t_1}\trans\tick 
\mathit{Sys_2}  \dots 
\trans{t_{n-1}}\trans\tick  \mathit{Sys_n}
\end{displaymath}
\end{small}%
such that the traces $t_j$ contain no $\tick$-actions, for any $j \in  1 .. n{-}1 $, and for any  $i \in  1 .. n $ $\mathit{Sys_i}= \confCPS {E_i}{P_i} $ with 
$E_i = \envCPS 
{\statefun^i{}} 
{\actuatorfun^i{}} 
{ \gamma }  
{\evolmap{}}
{ \epsilon }  
{\measmap{}}   
{\invariantfun{}}$.
Then, for any $i \in 1 .. n{-}1 $ we have the following:
\begin{itemize}[noitemsep]
\item  if   $ \actuatorfun^i{}(\mathit{cool})= \off $ then
 $\statefun^i{}(\mathit{temp})\in [0, 11.1+ \gamma ]$
 and $\statefun^i{}(\mathit{stress})=0$ if $ \statefun^i{}(\mathit{temp})  \in [0, 10.9+\gamma ] $ and, otherwise,  $\statefun^i{}(\mathit{stress})=1 $;

\item   if   $\actuatorfun^i{}(\mathit{cool})= \off $ and 
$\statefun^i{}(\mathit{temp})\in (10.1, 11.1+ \gamma ]$ then   $ \actuatorfun^{i+1}{}(\mathit{cool}) =\on$
 and $\statefun^{i{+}1}{}(\mathit{stress})  \in 1..2$;

\item  if  $ \actuatorfun^i{}(\mathit{cool})=\on$ then   $\statefun^i{}(\mathit{temp}) \in 
( 9.9-k *(1+\gamma), 11.1+\gamma -k*(1- \gamma)] $, 
for some  $k  \in 1..5$,   such that $\actuatorfun^{i-k}{}(\mathit{cool}) =\off $ and $ \actuatorfun^{i-j}{}(\mathit{cool}) =\on $, for $j \in 0..k{-}1$;
moreover,  if  $k\in 1.. 3$ then     $\statefun^i{}(\mathit{stress})  \in  1..k{+}1  $, otherwise,  
$\statefun^i{}(\mathit{stress}) =0$.

\end{itemize}
\end{lemma}
\begin{proof}
Similar to the proof of    \autoref{lem:sys}.
The crucial difference w.r.t.\ the proof of   \autoref{lem:sys} 
is limited to 
the second part of the third item. In particular the part saying that
$\statefun^i{}(\mathit{stress}) =0$, when $k \in 4..5$.
Now, after $3$ time units of cooling, the state variable 
$ \mathit{stress} $ lays  in the integer interval $  1..k{+}1=1..4$.
 Thus, in order to  have  $\statefun^i{}(\mathit{stress}) =0$, when $k \in 4..5$, the   temperature in the
third time slot of  the cooling must be less than or equal to $9.9$.
However, from the first statement of the third item 
 we deduce that,   in the
third time slot of cooling,  the  state variable $\mathit{temp}$
  reaches at most 
$ 11.1+\gamma -3*(1- \gamma)  =  8.1+4\gamma $. Thus, 
Hence we have that $ 8.1+4\gamma \leq 9.9$ for $\gamma \leq \frac{9}{20}$.
\end{proof}

The following lemma is a variant of \autoref{prop:sys}. 

\begin{lemma} 
\label{prop:sys:damage}
Let $\mathit{Sys}$ be the system defined in \autoref{exa:sys} and
$\gamma$ such that 
  $0.4 < \gamma \leq \frac{9}{20}$. 
If $\replaceENV {\mathit{Sys}}   \delta \gamma    \trans{t} Sys'$, for some $t=\alpha_1 \ldots \alpha_n$,  then
$\alpha_i \in \{ \tau , \tick \}$, for any $i \in 1 .. n$. 
\end{lemma}
\begin{proof}
By \autoref{lem:sys3}, the temperature will always lay in the real
interval $ [0, 11.1+ \gamma ]$. As a consequence, since $ \gamma \leq
\frac{9}{20}$, the system will never deadlock.

Moreover, after $5$ $\tick$ action of coolant the state variable
$\mathit{temp}$ is in $( 9.9-5 *(1+\gamma), 11.1+\gamma -5*(1- \gamma)]
=(4.9 -5\gamma \, , \, 6.1+6\gamma]$. Since $\epsilon = 0.1$, the value
detected from the sensor will be in the real interval $(4.8 -5\gamma \, ,
\, 6.2+6\gamma]$. Thus, the temperature sensed by $\mathit{IDS}$ will be
at most $6.2 + 6\gamma \leq 6.2+6*\frac{9}{20}\leq 10 $, and no alarm will
be fired.

 Finally, the maximum value that can be reached  by the state variable 
$ \mathit{stress} $ is $   k{+}1$m for $k=3$. As a consequence, 
 the system will not reach an unsafe state. 
\end{proof}

The following Lemma is a variant of \autoref{prop:X}.   Here the behaviour of $\mathit{Sys} $ is  parametric on the 
uncertainty. 

\begin{lemma}
\label{prop:X3}
Let $\mathit{Sys}$ be the system defined in \autoref{exa:sys} and
$\gamma$ such that  $0.4 < \gamma \leq \frac{9}{20}$. 
Then, for 
any execution trace of $\replaceENV  {Sys}  \delta \gamma$ we have the following:
\begin{itemize}[noitemsep]
\item if either process $\mathit{Ctrl}$ or process $\mathit{IDS}$ senses a temperature above $10$ then the value of
the state variable $\mathit{temp}$ ranges over $(9.9, 11.1+\gamma]$;
\item 
when the process  $\mathit{IDS}$  tests the temperature the value of
the state variable $\mathit{temp}$ 
ranges over $( 9.9-5 *(1+\gamma), 11.1+\gamma -5*(1- \gamma)] $. 
\end{itemize}
\end{lemma}
\begin{proof}
As to the first statement,  since $\epsilon=0.1$, if either process $\mathit{Ctrl}$ or process  $\mathit{IDS}$ senses  a temperature above $10$ then the value of
the state variable $\mathit{temp}$  is above $9.9$. 
By \autoref{lem:sys3}, the state variable $\mathit{temp}$ is less than or equal to $11.1+\gamma$.
Therefore, \emph{if either process $\mathit{Ctrl}$ or process  $\mathit{IDS}$ sense }  a temperature above $10$ then the value of
the state variable $\mathit{temp}$ is in $(9.9,11.1+\gamma]$.

Let us prove now the second statement. When the process $\mathit{IDS}$
tests the temperature then the coolant has been active for $5$ $\tick$
actions. By \autoref{lem:sys3}, the state variable $\mathit{temp}$ ranges
over $( 9.9-5 *(1+\gamma), 11.1+\gamma -5*(1- \gamma)] $. 
\end{proof}

Everything is finally in place to prove \autoref{prop:toll}.

\begin{proof}[Proof of \autoref{prop:toll}] 
For (1) we have to show that  $ \replaceENV  {\mathit{Sys}} \delta \gamma \, \sqsubseteq \,  \mathit{Sys}$, for $\gamma \in (\frac{8}{20} ,\frac{9}{20})$.  
 But this obviously holds by \autoref{prop:sys:damage}.  

As regards item (2), we have to prove that $ \replaceENV  {\mathit{Sys}} \delta \gamma \, \not \sqsubseteq \,  \mathit{Sys}$, for $\gamma
> \frac{ 9}{20} $. By \autoref{prop:sys} it is enough to show that
the system $\replaceENV {\mathit{Sys}}   \delta \gamma$ has a trace which either
(i) sends an alarm, or (ii) deadlocks, or (iii) enters in an unsafe state. We can easily build up a trace for
$\replaceENV {\mathit{Sys}}  {\delta} \gamma$ in which, after $10$
$\tick$-actions, in the $11$-th time slot, the value of the state
variable $\mathit{temp}$ is $10.1$. In fact, it is enough to increase the
temperature of $ 1.01$ degrees for the first $10$ rounds. Notice that this
is an admissible value since, $ 1.01 \in [ 1-\gamma,1+\gamma ]$, for any $
\gamma > \frac{ 9}{20}$. Being $10.1$ the value of the state variable
$\mathit{temp}$, there is an execution trace in which the sensed
temperature is $10$ (recall that $\epsilon=0.1$) and hence the cooling
system is not activated but the state variable $\mathit{stress}$ will be increased. 
In the following time slot, i.e.,
the $12$-th time slot, the temperature may reach at most the value
$10.1 + 1+\gamma$ and the state variable  $\mathit{stress}$ is $1$.  Now, if $10.1 + 1+\gamma>50$ then the system deadlocks.
Otherwise, the controller will activate the cooling system, and after $3$ time
units of cooling, in the $15$-th time slot, the state variable  $\mathit{stress}$ will be $4$ and the variable 
$\mathit{temp}$ will be at most $11.1+\gamma -3(1-\gamma)=8.1+4\gamma$.  
Thus, there is an execution trace in which the
 temperature is $ 8.1+4\gamma$, which will be greater than $9.9$ being
$\gamma> \frac{ 9}{20}$. As a consequence, in the 
next time slot, the state variable  $\mathit{stress}$ will be $5$
and the system will enter in an unsafe state.

This is enough to derive that $ \replaceENV  {\mathit{Sys}} \delta \gamma \, \not \sqsubseteq \,  \mathit{Sys}$, for $\gamma
> \frac{ 9}{20} $. 
\end{proof}

\begin{proof}[Proof of \autoref{thm:sound2}]
Consider the case of the definitive impact. By \autoref{lem:top}, if $M \parallel A\trans t$
then $ M \parallel \mathit{Top}(C)\trans t$. This entails $ M \parallel A \sqsubseteq M \parallel
\mathit{Top}(C) $. Thus, if $ M \parallel \mathit{Top}(C) \sqsubseteq
{\replaceENV M {\uncertaintyfun{}} {{\uncertaintyfun{}}{+}{\xi}}}$, for
$\xi \in \mathbb{R}^{\hat{\cal X}}$, $\xi >0$, then $ M \parallel A \sqsubseteq {\replaceENV M {\uncertaintyfun{}} {{\uncertaintyfun{}}{+}{\xi}}}$, by transitivity of
$\sqsubseteq$.

The proof in the case of the pointwise impact is by contradiction. Suppose
$\xi' > \xi $. Since $ \mathit{Top}(C) $ has a pointwise impact $\xi$ at
time $m$, it follows that $\xi$ is given by:

\begin{small}
\begin{center}
\begin{math}
  \inf \big\{ \xi''  :  \xi'' {\in} \mathbb{R}^{\hat{\cal X}} 
\: \wedge \: M \parallel   \mathit{Top}(C)    \sqsubseteq_{m ..n} 
\replaceENV M  {\uncertaintyfun{}}  {{\uncertaintyfun{}} {+} {\xi''}},  n {\in} \mathbb{N} {\cup} \infty   \big\}.
\end{math}
\end{center}
\end{small}%
Similarly, since $A$ has a pointwise impact $\xi'$ at time $m'$, it
follows that $\xi'$ is given by

\begin{small}
\begin{center}
\begin{math}
 \inf \big\{ \xi''  :  \xi'' {\in} \mathbb{R}^{\hat{\cal X}} 
\,  \wedge \, M \parallel  A   \sqsubseteq_{m'..n} 
\replaceENV M  {\uncertaintyfun{}}  {{\uncertaintyfun{}} {+} {\xi''}},  n {\in} \mathbb{N} {\cup} \infty   \big\}.
\end{math}
\end{center}
\end{small}

Now, if it were $m=m'$ then $\xi \geq \xi'$ 
because $ M \parallel   A  \trans{t}$ entails 
$ M \parallel \mathit{Top}(C)\trans t$., 
 by an application of \autoref{lem:top}. 
This is contradiction with the fact that $\xi < \xi'$,
Thus, it must be $m' < m$.
Now, since both $\xi $ and  $\xi'$ are the infimum functions and since $\xi' > \xi $,  there exist   $\overline{\xi}$ 
and $\overline{ \xi'}$  such that   $\xi \leq \overline{\xi}\leq \xi' \leq \overline{ \xi'}$ and 
 $ M \parallel   \mathit{Top}(C)    \sqsubseteq_{m..n} 
\replaceENV M  {\uncertaintyfun{}}  {{\uncertaintyfun{}} {+} {\overline{\xi}}}$, for some $n$,
and 
 $ M \parallel  A   \sqsubseteq_{m'..n'} 
\replaceENV M  {\uncertaintyfun{}}  {{\uncertaintyfun{}} {+} {\overline{\xi'}}}$, for some $n'$.

Hence, from  $ M \parallel  A   \sqsubseteq_{m'..n'} 
\replaceENV M  {\uncertaintyfun{}}  {{\uncertaintyfun{}} {+} {\overline{\xi'}}}$,
we have that   there exists a trace $t$ with $\#\tick(t)=m'-1$ such that 
$ M \parallel   A  \trans{t}$ and $\replaceENV M  {\uncertaintyfun{}}  {{\uncertaintyfun{}} {+} {\overline{\xi'}}}  \not\!\!\Trans{\hat{t}}$. 
Since $\overline{\xi} \leq \overline{\xi'} $,   by monotonicity (\autoref{prop:monotonicity}), we deduce that
  $\replaceENV M  {\uncertaintyfun{}}  {{\uncertaintyfun{}} {+} {\overline \xi}}  \not\!\!\Trans{\hat{t}}$. 
  Moreover,   by \autoref{lem:top}, $ M \parallel   A  \trans{t}$ entails 
$ M \parallel \mathit{Top}(C)\trans t$. 

Summarising, there exists a trace $t$ with $\#\tick(t)=m'-1$ such that $ M
\parallel \mathit{Top}(C) \trans{t}$ and $\replaceENV M
{\uncertaintyfun{}} {{\uncertaintyfun{}} {+} {\overline \xi }}
\not\!\!\Trans{\hat{t}}$. However, this fact and $m' < m$ is in
contradiction with $M \parallel \mathit{Top}(C) \sqsubseteq_{m ..n}
\replaceENV M {\uncertaintyfun{}} {{\uncertaintyfun{}} {+} {\overline \xi
}}$, for some $n$.

This is enough to derive the statement. 
\end{proof}

\begin{proof}[Proof of \autoref{prop:effect2}]
Let us prove the first sub-result. 
As demonstrated  in \autoref{exa:att:DoS2},  we know that 
$\mathit{Sys} \parallel A  \sqsubseteq_{14..\infty}  \mathit{Sys}$
because  in the $14$-th time slot 
the compound system  will violate the safety conditions emitting an $\unsafe$-action  until the invariant will  be violated.
No alarm will be emitted.

Since the system keeps violating the safety condition
the temperature must remain greater than $9.9$.
As proved for \autoref{lem:sys3} 
 we can prove that we have that the temperature is less than or equal to $10.1+\gamma $.  
Hence, in the time slot before getting in deadlock, the temperature 
of the system is in the real interval $(9.9,10.1+\gamma]$.
To deadlock  with one $\tick$ action and from  a temperature in the real interval $(9.9,10.1+\gamma]$, either the temperature  reaches  a value greater than $50$ (namely, $10.1+\gamma+1+\gamma > 50$) or 
 the temperature reaches a value less than $ 0$ (namely, $9.9-1-\gamma <  0$ ). 
Since $\gamma \leq 8.9$, both cases can not occur. Thus,  we have that 
\[
\mathit{Sys}  \parallel  A   \not \sqsubseteq \,  \replaceENV{\mathit{Sys}}  {\delta} {\gamma} 
\enspace .   
\]
Let us prove the second sub-result. 
That is,  \[ Sys  \parallel  A    \sqsubseteq \,   \replaceENV{\mathit{Sys}}  {\delta} {\gamma} \]
 for $\gamma >8.9$. 
 We demonstrate that 
whenever $\mathit{Sys} \parallel  A \trans{t}$, for some trace $t$, then 
$\replaceENV {\mathit{Sys}  } \delta \gamma\Trans{\hat t}$ as well. 
We will proceed by case analysis on the kind of actions contained in $t$. 
We distinguish three possible cases.

\begin{itemize}[noitemsep]

\item The trace $t$ contains contains only $\tau$-, $\tick$-, $\unsafe$-
and $\dead$-actions. As discussed in \autoref{exa:att:DoS2},
$\mathit{Sys} \parallel A \; \sqsubseteq_{14..\infty} \; \mathit{Sys}$
because in the $14$-th time slot the system will violate the safety
conditions emitting an $\unsafe$-action until the invariant will be
broken. No alarm will be emitted. Note that, when $\mathit{Sys} \parallel
A$ enters in an unsafe state then the temperature is at most
$9.9+(1+\delta)+5(1+\delta)=9.9+6(1.4)=18.3$. Moreover, the fastest
execution trace, reaching an unsafe state, deadlocks just after $\lceil
\frac{ 50-18.3}{1 + \delta } \rceil = \lceil \frac{ 31,7}{1.4 } \rceil=23$
$\tick$-actions. Hence, there are $m,n \in \mathbb{N}$, with $m\geq 14$ and
$n\geq m+23$, such that the trace $t$ of $\mathit{Sys} \parallel A $
satisfies the following conditions: (i) in the time interval $1..m-1$ the
trace $t$ of is composed by $\tau$- and $\tick$-actions; (ii) in the time
interval $m..(n-1)$, the trace $t$ is composed by $\tau$-, $\tick$- and
$\unsafe$- actions; in the $n$-th time slot the trace $t$ deadlocks.

By monotonicity (\autoref{prop:monotonicity}), it is enough 
to show that such a trace exists 
for $\replaceENV {\mathit{Sys}  } \delta \gamma$, with 
 $8.9 < \gamma < 9$. In fact, if this trace exists for  $ 8.9 < \gamma < 9$, then it would also exist for
 $  \gamma \geq 9$. 
In the following, we show how to build the trace of $\replaceENV {\mathit{Sys}  } \delta \gamma$ which simulates  the trace $t$ of $ \mathit{Sys}  \parallel A$. 
We build up the trace in three steps: (i)  
the sub-trace from time slot $1$ to  time slot $m{-}6$;
(ii) the sub-trace from the time slot $m{-}5$ to the time slot $n{-}1$;
(iii) the final part of the trace reaching the deadlock.
\begin{itemize}
\item[(i)] 
As $\gamma>8.9$ (and hence $1+\gamma>9.9$), the system may increment the
temperature of $9.9$ degrees after a single $\tick$-action. Hence, we
choose the trace in which the system $\replaceENV {\mathit{Sys} } \delta
\gamma$, in the second time slot, reaches the temperature equal to $9.9$.
Moreover, the system may maintain this temperature value until the
$(m{-}6)$-th time slot (indeed $0$ is an admissible increasing since $0
\in [1-\gamma,1+\gamma]\supseteq [-7.9,10.9]$) . Obviously, with a
temperature equal to $9.9$, only $\tau$- and $\tick$-actions are possible.

\item[(ii)]
Let $k \in \mathbb{R}$ such that $0< k < \gamma-8.9  $ (such $k$ exists since $\gamma>8.9 $).  
We may consider an increment of the temperature of $k$. 
This implies that in the $(m{-}5)$-th time slot, the system 
 $\replaceENV {\mathit{Sys}  } \delta \gamma$ may reach the temperature $9.9+k$. 
 Note that $ k$ is an admissible increment since $0< k < \gamma-8.9  $ and $8.9 < \gamma < 9$ entails $k \in (0,0.1)$. 
 Moreover, the system may maintain this temperature value  until the $(n{-}1)$-th time slot
(indeed, as said before, $0$ is an admissible increment).
Summarising from the $(m{-}5)$-th time slot to the $(n{-}1)$-th time slot,
the temperature may remain equal to $9.9+k \in (9.9,10)$. As a
consequence, from the $m$-th time slot to the $(n{-}1)$-th time slot the
system $\replaceENV {\mathit{Sys} } \delta \gamma$ may enter in an unsafe
state (i.e., $\safefun{}(E)=\false$). Thus, an $\unsafe$-action may be
performed in the time interval $m..(n{-}1)$. Moreover, since
$\epsilon=0.1$ and the temperature is e $9.9+k \in (9.9,10)$, we can
always assume that the cooling is not activated until the $(n{-}1)$-th
time slot. This implies that neither alarm nor deadlock occur.

\item[(iii)]
At this point, since in the $(n{-}1)$-th time slot the temperature is
equal to $9.9 + k \in (9.9,10)$ (recall that $k \in (0,1)$), the cooling
may be activated. We may consider a decrement of $1+\gamma$. In this
manner, in the $n$-th time slot the system may reach a temperature of
$9.9+k-(1+\gamma)< 9.9+0 -1 -8.9 =0$ degrees, and the system $\replaceENV
{\mathit{Sys} } \delta \gamma$ will deadlock.

\end{itemize}
 
Summarising, for any $\gamma > 8.9 $ the system $\replaceENV {\mathit{Sys}
} \delta \gamma$ can mimic any trace $t$ of $ \mathit{Sys} \parallel A$.

\item The trace $t$ contains contains only $\tau$-, $\tick$- and
$\unsafe$-actions. This case is similar to the previous one.

\item The trace $t$ contains only $\tau$-, $\tick$- and
$\overline{alarm}$-actions. This case cannot occur. In fact, as discussed
in \autoref{exa:att:DoS2}, the process $\mathit{Ctrl}$ never activates the
$\mathit{Cooling}$ component (and hence also the $\mathit{IDS}$ component,
which is the only one that could send an alarm) since it will always
detect a temperature below $10$.
 
\item The trace $t$ contains only $\tau$- and $\tick$-actions. If the
system $\mathit{Sys} \parallel A $ has a trace $t$ that contains only
$\tau$- and $\tick$-actions, then, by \autoref{prop:sys}, the system
$\mathit{Sys}$ in isolation must have a similar trace with the same number
of $\tick$-actions. By an application of \autoref{prop:monotonicity}, as
$\delta<\gamma$, any trace of $\mathit{Sys} $ can be simulated by
$\replaceENV {\mathit{Sys} } \delta \gamma$. As a consequence,
$\replaceENV {\mathit{Sys} } \delta \gamma\Trans{\hat t}$.
\end{itemize}

This is enough to derive that:
\[
\mathit{Sys}   \parallel  A    \sqsubseteq \,   \replaceENV {\mathit{Sys}  } \delta \gamma
\enspace , 
\]
which concludes the proof.
\end{proof}

\end{document}